\documentclass[11pt,a4paper]{amsart}
\usepackage{amssymb}
\usepackage{cite,scrtime}
\usepackage{amsmath,dsfont}
\usepackage{mathrsfs}
\usepackage{fancyhdr}
\usepackage{verbatim}
\usepackage{boxedminipage}
\usepackage{colortbl}
\usepackage{algorithm}
\usepackage{algorithmic}

\newcommand{\tw}{{\mathbf{tw}}}
\newcommand{\cc}{{\mathbf{cc}}}
\newcommand{\mc}{\mathcal}
\newcommand{\mMSO}{\mbox{$p$-min-MSO}}

\usepackage{graphicx}

\newcommand{\cal}{\mathcal}

\renewcommand{\deg}{\text{deg}}

\def\debut{\begin{itemize}\item[{\bf [[}]\small}
\def\term{\hfill {\bf]]} \end{itemize} }

\newtheorem{claimN}{Claim}

\theoremstyle{plain}
\newtheorem{theorem}{Theorem}[section]
\newtheorem{lemma}[theorem]{Lemma}
\newtheorem{corollary}[theorem]{Corollary}
\newtheorem{proposition}[theorem]{Proposition}

\theoremstyle{definition}

\title{Hitting and Harvesting Pumpkins}
\thanks{A preliminary conference version of this work appeared in the \emph{Proceedings of the 19th Annual European
Symposium on Algorithms (ESA), volume 6942 of LNCS, pages 394-407,
Saarbr\"{u}cken, Germany, September 2011}. This work was supported in part by
AGAPE (ANR-09-BLAN-0159) and GRATOS (ANR-09-JCJC-0041-01) French projects, by a
project funded by DAE (India), and by the Actions de Recherche Concert\'ees
(ARC) fund of the Communaut\'e fran\c{c}aise de Belgique (Belgium). Gwena\"el
Joret is a Postdoctoral Researcher of the Fonds National de la Recherche
Scientifique (F.R.S.--FNRS)}
\author[G. Joret, C. Paul, I. Sau, S. Saurabh, and S. Thomass\'e]{Gwena\"{e}l Joret, Christophe Paul, Ignasi Sau,\\Saket Saurabh, and St\'ephan Thomass\'e}

\address{\newline D\'epartement d'Informatique
\newline Universit\'e Libre de Bruxelles
\newline Brussels, Belgium}
\email{gjoret@ulb.ac.be}

\address{\newline AlGCo project-team
\newline CNRS, LIRMM
\newline Montpellier, France}
\email{\{paul,sau\}@lirmm.fr}

\address{\newline Laboratoire LIP
\newline Universit\'e de Lyon, CNRS, ENS Lyon, INRIA, UCBL
\newline Lyon, France}
\email{stephan.thomasse@ens-lyon.fr}

\address{\newline The Institute of Mathematical Sciences
\newline Chennai, India}
\email{saket@imsc.res.in}

\begin{document}

\sloppy \maketitle

\begin{abstract}
The \emph{$c$-pumpkin} is the graph with two vertices linked by $c \geq 1$
parallel edges. A \emph{$c$-pumpkin-model} in a graph $G$ is a pair $\{A, B\}$
of disjoint subsets of vertices of $G$, each inducing a connected subgraph of
$G$, such that there are at least $c$ edges in $G$ between $A$ and $B$. We
focus on hitting and packing $c$-pumpkin-models in a given graph in the realm
of approximation algorithms and parameterized algorithms. We give an FPT
algorithm running in time $2^{\mathcal{O}(k)} n^{\mathcal{O}(1)}$ deciding, for
any fixed $c \geq 1$, whether all $c$-pumpkin-models can be hit by at most $k$
vertices. This generalizes known single-exponential FPT algorithms for
\textsc{Vertex Cover} and \textsc{Feedback Vertex Set}, which correspond to the
cases $c=1,2$ respectively. Finally, we present an $\mathcal{O}(\log
n)$-approximation algorithm for both the problems of hitting all
$c$-pumpkin-models with a smallest number of vertices, and packing a maximum
number of vertex-disjoint $c$-pumpkin-models.

\vspace{0.25cm} \textbf{Keywords:} Hitting and packing; parameterized
complexity; approximation algorithm; single-exponential algorithm; iterative
compression; graph minors.
\end{abstract}


\vspace{.3cm}


\section{Introduction}
\label{sec:intro}

The \emph{$c$-pumpkin} is the graph with two vertices linked by $c \geq 1$
parallel edges. A \emph{$c$-pumpkin-model} in a graph $G$ is a pair $\{A, B\}$
of disjoint subsets of vertices of $G$, each inducing a connected subgraph of
$G$, such that there are at least $c$ edges in $G$ between $A$ and $B$. In this
article we study the problems of hitting all $c$-pumpkin-models of a given
graph $G$ with as few vertices as possible, and packing as many disjoint
$c$-pumpkin-models in $G$ as possible. As discussed below, these problems
generalize several well-studied problems in algorithmic graph theory. We focus
on FPT algorithms for the parameterized version of the hitting problem, as well
as on polynomial-time approximation algorithms for the optimization version of
both the packing and hitting problems.

\subsection*{FPT algorithms}
From the parameterized complexity perspective, we study the following problem
for every fixed integer $c \geq 1$.\vspace{.2cm}

\begin{center}
\begin{boxedminipage}{.8\textwidth}

$p$-$c$-\textsc{Pumpkin}-\textsc{Hitting} ($p$-$c$-\textsc{Hit})

\begin{tabular}{ r l }
\textit{~~~~Instance:} & A graph $G$  and a non-negative integer $k$. \\
\textit{Parameter:} & $k$.\\
\textit{Question:} & Does there exist $S \subseteq V(G)$, $|S| \leq k$, such that\\
~~~~~~~~~~~~~~~~~~ & $G \setminus S$ does not contain the $c$-pumpkin as a
minor?
 \\
\end{tabular}
\end{boxedminipage}

\end{center}\vspace{.2cm}

Several special cases of this problem are well studied in the parameterized
complexity.
When $c=1$, the $p$-$c$-\textsc{Hit} problem is the {\sc $p$-Vertex Cover}
problem~\cite{BFR98,CKX10}. For $c=2$, it is the  {\sc $p$-Feedback Vertex Set}
problem~\cite{GGH+06,DFL+07,CaoCL10}. When $c=3$, this corresponds to the
recently introduced {\sc $p$-Diamond Hitting Set} problem~\cite{FJP10}.

The $p$-$c$-\textsc{Hit} problem can also be seen as a particular case of the
following problem, recently introduced by Fomin \emph{et al}.~\cite{FLM+11} and
studied from the kernelization perspective: Let $\cal F$ be a finite set of
graphs. In the
$p$-$\mathcal{F}$-\textsc{Hit} problem, 
we are given an $n$-vertex graph $G$ and an integer $k$ as input, and asked
whether at most $k$ vertices can be deleted from $G$ such that the resulting
graph does not contain any graph from  ${\cal F}$ as a minor.
Among other results, it is proved in~\cite{FLM+11} that if $\mathcal{F}$
contains a $c$-pumpkin for some $c \geq 1$, then $p$-$\mathcal{F}$-\textsc{Hit}
admits a kernel of size $\mathcal{O}(k^2 \log^{3/2}k)$. As discussed in
Section~\ref{sec:algorithm}, this kernel leads to a simple FPT algorithm for
$p$-$\mathcal{F}$-\textsc{Hit} in this case, and in particular for
$p$-$c$-\textsc{Hit}, with running time $2^{\mathcal{O}(k \log k)} \cdot
n^{\mathcal{O}(1)}$. A natural question is whether there exists an algorithm
for $p$-$c$-\textsc{Hit} with running time $2^{\mathcal{O}(k)} \cdot
n^{\mathcal{O}(1)}$ for every fixed $c \geq 1$. Such algorithms are called {\sl
single-exponential}.
For the \textsc{$p$-Vertex Cover} problem the existence of single-exponential
algorithms is well-known since almost the beginnings of the field of
Parameterized Complexity~\cite{BFR98}, the best current algorithm being by Chen
\emph{et al}.~\cite{CKX10}. On the other hand, the question about the existence
of single-exponential algorithms for \textsc{$p$-Feedback Vertex Set} was open
for a long time and was finally settled independently by Guo \emph{et
al}.~\cite{GGH+06} (using iterative compression) and by Dehne \emph{et
al}.~\cite{DFL+07}. The current champion {\sl deterministic} algorithm for \textsc{$p$-Feedback
Vertex Set} runs in time $\mathcal{O}(3.83^kk\cdot n^2)$~\cite{CaoCL10}, whereas the fastest {\sl randomized} one runs in time $\mathcal{O}(3^k) \cdot \text{poly}(n)$~\cite{CNP+11}.

We present in Section~\ref{sec:algorithm} a single-exponential algorithm for
$p$-$c$-\textsc{Hit} for every fixed $c \geq 1$, using a combination of a
kernelization-like technique and iterative compression. (In fact, we will solve the {\sl constructive} version of the problem.) Notice that this
generalizes the above results for \textsc{$p$-Vertex Cover} and
\textsc{$p$-Feedback Vertex Set}.  We remark that asymptotically these
algorithms are optimal, that is, it is known that unless ETH fails neither
\textsc{$p$-Vertex Cover} nor \textsc{$p$-Feedback Vertex Set} admit an
algorithm with running time $2^{o(k)}\cdot
n^{\mathcal{O}(1)}$~\cite{CCF+05,ImpagliazzoPZ01}. It is worth mentioning here
that a similar quantitative approach was taken by Lampis~\cite{Lampis10} for
graph problems expressible in MSOL parameterized by the sum of the formula size
and the size of a minimum vertex cover of the input graph.

\subsection*{Approximation algorithms}
For a fixed integer $c \geq 1$, we define the following two optimization
problems.\vspace{.2cm}
\begin{center}
\begin{boxedminipage}{.8\textwidth}

\textsc{Minimum} $c$-\textsc{Pumpkin}-\textsc{Hitting} (\textsc{Min}
$c$-\textsc{Hit})

\begin{tabular}{ r l }
\textit{~~~~Input:} & A graph $G$. \\
\textit{Output:} & A subset $S \subseteq V(G)$ such that $G \setminus S$\\
~~~~~~~~~~~~~~~~~~ &  does not contain the $c$-pumpkin as a minor.\\
\textit{Objective:} & Minimize $|S|$.\\
\end{tabular}
\end{boxedminipage}
\end{center}\vspace{.1cm}

\begin{center}
\begin{boxedminipage}{.8\textwidth}

\textsc{Maximum} $c$-\textsc{Pumpkin}-\textsc{Packing} (\textsc{Max}
$c$-\textsc{Pack})

\begin{tabular}{ r l }
\textit{~~~~Input:} & A graph $G$. \\
\textit{Output:} & A collection $\mathcal{M}$ of vertex-disjoint subgraphs of $G$,\\
~~~~~~~~~~~~~~~~~~ &  each containing the $c$-pumpkin as a minor.\\
\textit{Objective:} & Maximize $|\mathcal{M}|$.\\
\end{tabular}
\end{boxedminipage}
\end{center}\vspace{.2cm}

Let us now discuss how the above problems encompass several well-known
problems. For $c=1$, \textsc{Min} $1$-\textsc{Hit} is the \textsc{Minimum
Vertex Cover} problem, which can be easily 2-approximated by finding any
maximal matching, whereas \textsc{Max} $1$-\textsc{Pack} corresponds to finding
a \textsc{Maximum Matching}, which can be done in polynomial time. For $c=2$,
\textsc{Min} $2$-\textsc{Hit} is the \textsc{Minimum Feedback Vertex Set}
problem, which can be also 2-approximated~\cite{BBF99,BeGe96}, whereas
\textsc{Max} $2$-\textsc{Pack} corresponds to \textsc{Maximum Vertex-Disjoint
Cycle Packing}, which can be approximated to within a $\mathcal{O}(\log n)$
factor~\cite{KNS+07}. For $c=3$, \textsc{Min} $3$-\textsc{Hit} is the
\textsc{Diamond Hitting Set} problem studied by Fiorini \emph{et al.}
in~\cite{FJP10}, where a $9$-approximation algorithm is given.

We provide in Section~\ref{sec:combinatorial} an algorithm that approximates
both the \textsc{Min} $c$-\textsc{Hit} and the \textsc{Max} $c$-\textsc{Pack}
problems to within a $\mathcal{O}(\log n)$ factor for every fixed $c \geq 1$.
Note that this algorithm matches the best existing algorithms for \textsc{Max}
$c$-\textsc{Pack} for $c=2$~\cite{KNS+07}. For the \textsc{Min}
$c$-\textsc{Hit} problem, our result is only a slight improvement on the
$\mathcal{O}(\log^{3/2} n)$-approximation algorithm given in~\cite{FLM+11}.
However, for the \textsc{Max} $c$-\textsc{Pack} problem, there was no
approximation algorithm known before except for the $c\leq 2$ case. Also, let
us remark that, for $c \geq2$ and every fixed $\varepsilon > 0$, \textsc{Max}
$c$-\textsc{Pack} is quasi-NP-hard to approximate to within a
$\mathcal{O}(\log^{1/2 - \varepsilon} n)$ factor. For $c=2$ this was shown by
Friggstad and Salavatipour~\cite{FrSa11}, and their result can be extended to
the case $c>2$ in the following straightforward way. Given an instance $G$ of
\textsc{Max} $2$-\textsc{Pack}, we build an instance of \textsc{Max}
$c$-\textsc{Pack} by replacing each edge of $G$ with $c-1$ parallel edges. This
approximation preserving transformation shows that \textsc{Max}
$c$-\textsc{Pack} is quasi-NP-hard to approximate to within a
$\mathcal{O}(\log^{1/2 - \varepsilon} n)$ factor for any $c\geq 2$.



The main ingredient of our approximation algorithm is the following
combinatorial result: We show that every $n$-vertex graph $G$ either contains a
small $c$-pumpkin-model or has a structure that can be reduced in polynomial
time, giving a smaller equivalent instance for both the \textsc{Min}
$c$-\textsc{Hit} and the \textsc{Max} $c$-\textsc{Pack} problems. Here by a
``small'' $c$-pumpkin-model, we mean a model of size at most $f(c)\cdot \log n$
for some function $f$ independent of $n$. This result extends one of Fiorini
{\it et al.}~\cite{FJP10}, who dealt with the case $c=3$.




\section{Preliminaries}
\label{sec:prelim}


\subsection*{Graphs}
We use standard graph terminology, see for instance~\cite{Die05}. All graphs in
this article are finite and undirected, and may have parallel edges but no
loops. We will sometimes restrict our attention to simple graphs, that is,
graphs without parallel edges.

Given a graph $G$, we denote the vertex set of $G$ by $V(G)$ and the edge set
of $G$ by $E(G)$. We use the shorthand $|G|$ for the number of vertices in $G$.
For a subset $X\subseteq V(G)$, we use $G[X]$ to denote the subgraph of $G$
induced by $X$. For a subset $Y\subseteq E(G)$ we let $G[Y]$ be the graph with
$E(G[Y]):=Y$ and with $V(G[Y])$ being the set of vertices of $G$ incident with
some edge in $Y$. For a subset $X \subseteq V(G)$, we may use the notation $G
\setminus X$ to denote the graph $G[V(G)\setminus X]$.

The set of neighbors of a vertex $v$ of a graph $G$ is denoted by $N_{G}(v)$.
The \emph{degree} $\deg_{G}(v)$ of a vertex $v\in V(G)$ is defined as the
number of edges incident with $v$ (thus parallel edges are counted). We write
$\deg^{*}_{G}(v)$ for the number of neighbors of $v$, that is, $\deg^{*}_{G}(v)
:= |N_G(v)|$. Similarly, given a subgraph $H \subseteq G$ with $v \in V(H)$, we
can define in the natural way $N_{H}(v)$, $\deg_{H}(v)$, and $\deg^{*}_{H}(v)$,
that is, $N_{H}(v)=N_{G}(v)\cap V(H)$, $\deg_{H}(v)$ is the number of edges
incident with $v$ with both endpoints in $H$, and $\deg^{*}_{H}(v)=|N_{H}(v)|$.
In these notations, we may drop the subscript if the graph is clear from the
context. By the \emph{neighbors of a subgraph} $H \subseteq G$ we mean the set
of vertices in $V(G) \setminus V(H)$ that have at least one neighbor in $H$.
The minimum degree of a vertex in a graph $G$ is denoted $\delta(G)$, and the
maximum degree of a vertex in $G$ is denoted $\Delta(G)$. We use the notation
$\cc(G)$ to denote the number of connected components of $G$. Also, we let
$\mu(G)$ denote the maximum multiplicity of an edge in $G$. A graph is said to
be a {\em multipath} if its underlying simple graph (without parallel edges)
is isomorphic to a path.

\subsection*{Minors and models}
Given  a graph $G$ and an edge $e\in E(G)$, let $G\backslash e$ be the graph
obtained from $G$ by removing the edge $e$, and let $G\slash e$ be the graph
obtained from $G$ by contracting $e$ (we note that parallel edges resulting
from the contraction are kept but self loops are deleted). If $H$ can be obtained
from a subgraph of $G$ by a (possibly empty) sequence of edge contractions, we
say that $H$ is a \emph{minor} of $G$, and we denote it by $H \preceq_m G$. A
graph $G$ is {\em $H$-minor-free}, or simply {\em $H$-free}, if $G$ does not
contain $H$ as a minor. A \emph{model} of a graph $H$, or simply
\emph{$H$-model}, in a graph $G$ is a collection $\{S_{v} \subseteq V(G) \mid v
\in V(H)\}$ such that
\begin{itemize}
\item[(i)] $G[S_{v}]$ is connected for every $v\in V(H)$;\vspace{.4mm}
\item[(ii)] $S_{v}$ and $S_{w}$ are disjoint for every two distinct vertices $v, w$ of
$H$; and\vspace{.4mm}
\item[(iii)] there are at least as many edges between $S_{v}$ and $S_{w}$ in $G$
as between $v$ and $w$ in $H$, for every $vw \in E(H)$.
\end{itemize}
The {\em size} of the model is defined as $\sum_{v \in V(H)} |S_{v}|$. Clearly,
$H$ is a minor of $G$ if and only if there exists a model of $H$ in $G$. In
this paper, we will almost exclusively consider $H$-models with $H$ being
isomorphic to a $c$-pumpkin for some $c\geq 1$. Thus a $c$-pumpkin-model in a
graph $G$ is specified by an unordered pair $\{A, B\}$ of disjoint subsets of
vertices of $G$, each inducing a connected subgraph of $G$, such that there are
at least $c$ edges in $G$ between $A$ and $B$. A $c$-pumpkin-model $\{A, B\}$
of $G$ is said to be {\em minimal} if there is no $c$-pumpkin-model $\{A',
B'\}$ of $G$ with $A' \subseteq A$, $B' \subseteq B$, and $|A'| + |B'| < |A| +
|B|$.

A subset $X$ of vertices of a graph $G$ such that $G \setminus X$ has no
$c$-pumpkin-minor is called a {\em $c$-pumpkin-hitting set}, or simply {\em
$c$-hitting set}.
We denote by $\tau_{c}(G)$ the minimum size of a $c$-pumpkin-hitting set in
$G$. A collection $\mathcal{M}$ of vertex-disjoint subgraphs of a graph $G$,
each containing a $c$-pumpkin-model, is called a {\em $c$-pumpkin-packing}, or
simply {\em $c$-packing}. We denote by $\nu_{c}(G)$ the maximum size of a
$c$-pumpkin-packing in $G$. Obviously, for any graph $G$ it holds that
$\nu_{c}(G) \leq \tau_{c}(G)$, but the converse is not necessarily true.

The following lemma on models will be useful in our algorithms. The proof is
straightforward and hence is omitted.

\begin{lemma}
\label{lem:models_diameter} Suppose $G'$ is obtained from a graph $G$ by
contracting some vertex-disjoint subgraphs of $G$, each of diameter at most
$k$. Then, given an $H$-model in $G'$ of size $s$, one can compute in
polynomial time an $H$-model in $G$ of size at most $k\cdot \Delta(H)\cdot s$.
\end{lemma}


\subsection*{Parameterized algorithms}
 A \emph{parameterized problem} $\Pi$ is a subset of
$\Gamma^{*}\times \mathbb{N}$ for some finite alphabet $\Gamma$. An instance of
a parameterized problem consists of a pair $(x,k)$, where $k$ is called the
\emph{parameter}. A central notion in parameterized complexity is {\em fixed
parameter tractability (FPT)}, which means, for a given instance $(x,k)$,
solvability in time $f(k)\cdot p(|x|)$, where $f$ is some computable function of
$k$ and $p$ is a polynomial in the input size.

A \emph{kernelization algorithm} or, in short, a \emph{kernel} for a
parameterized problem $\Pi\subseteq \Gamma^{*}\times \mathbb{N}$ is an
algorithm that given $(x,k)\in \Gamma^{*}\times \mathbb{N} $ outputs in time
polynomial in $|x|+k$ a pair $(x',k')\in \Gamma^{*}\times \mathbb{N}$ such that
\begin{itemize}
\item[(i)] $(x,k)\in \Pi$
if and only if $(x',k')\in \Pi$; and\vspace{.4mm}
\item[(ii)] $\max\{|x'|, k' \}\leq g(k)$,
\end{itemize}
 where $g$ is some computable function. The function
$g$ is referred to as the \emph{size} of the kernel.
If~$g(k)=k^{\mathcal{O}(1)}$ or $g(k)=\mathcal{O}(k)$,
then we say that $\Pi$ admits a polynomial kernel and a linear kernel,
respectively.


\emph{Iterative compression} is a tool that has been used successfully in
finding fast FPT algorithms for a number of parameterized problems. The main
idea behind iterative compression is an algorithm which, given a solution of
size $k+1$ for a problem, either compresses it to a solution of size $k$ or
proves that there is no solution of size $k$. This technique was first
introduced by Reed \emph{et al}. to solve the \textsc{Odd Cycle Transversal}
problem~\cite{RSV04}, where one is interested in finding a set of at most $k$
vertices whose deletion makes the graph bipartite~\cite{RSV04}. Since then, it
has been extensively used in the literature, see for
instance~\cite{FGK+10,CLL+10,GGH+06}.

See~\cite{DowneyF99} for detailed introduction to Parameterized
Complexity.

\subsection*{Tree-width}
We briefly recall the definition of the tree-width of a graph. A \emph{tree
decomposition} of a graph $G$ is an ordered pair $(T, \{W_{x} \mid x \in
V(T)\})$ where $T$ is a tree and $\{W_{x} \mid x \in V(T)\}$ a family of
subsets of $V(G)$ (called {\em bags}) such that 
\begin{itemize}
\item[(i)] $\bigcup_{x \in V(T)} W_{x} = V(G)$;
\item[(ii)] for every edge $uv \in E(G)$, there exists $x \in V(T)$ with $u,v  \in
W_{x}$; and
\item[(iii)] for every vertex $u\in V(G)$, the set $\{x\in V(T) \mid u \in W_{x}\}$
induces a subtree of $T$.
\end{itemize}
The {\em width} of tree decomposition $(T,  \{W_{x} \mid x \in V(T)\})$ is
$\max \{ |W_{x}| - 1 \mid x \in V(T)\}$. The {\em tree-width} $\tw(G)$ of $G$
is the minimum width among all tree decompositions of $G$. We refer the reader
to Diestel's book~\cite{Die05} for an introduction to the theory of tree-width.
It is an easy exercise to check that the tree-width of a simple graph is an
upper bound on its minimum degree. This implies the following lemma.

\begin{lemma}
\label{lem:edges_tw} Every $n$-vertex simple graph with tree-width $k$ has at
most $k \cdot n$ edges.
\end{lemma}

We will need the following result of Bodlaender \emph{et
al}.~\cite{BodlaenderLTT97}.

\begin{theorem}[Bodlaender \emph{et al}.~\cite{BodlaenderLTT97}]
\label{cor:ExcludeTheta} Every graph not containing a $c$-pumpkin as a minor
has tree-width at most $2c-1$.
\end{theorem}

The following corollary is an immediate consequences of the above theorem.

\begin{corollary}
\label{cor:MultiExcludeTheta} Every $n$-vertex graph (may contain parallel edges) with no minor isomorphic
to a $c$-pumpkin has at most $(c-1)\cdot (2c-1) \cdot n$ edges.
\end{corollary}

Note that the
existence of a $c$-pumpkin-minor in a graph can be tested in polynomial time
by using the polynomial-time algorithm of Robertson and Seymour~\cite{RS95}. The following proposition states that $c$-pumpkin-minors can be found in {\sl linear time}.

\begin{proposition}
\label{prop:FindingTheta} For each fixed integer positive integer $c$, the existence of a $c$-pumpkin-minor in an $n$-vertex graph $G$ can be done in time $\mathcal{O}(n)$.
\end{proposition}
\begin{proof} We first check whether the treewidth of $G$ is at most $2c-1$, by using the linear-time algorithm of Bodlaender~\cite{Bod96}. If the treewidth of $G$ is strictly larger than $2c-1$, then by Theorem~\ref{cor:ExcludeTheta} we can conclude that $G$ contains a $c$-pumpkin-minor. Otherwise, the treewidth of $G$ is bounded, and we can test for the existence of a $c$-pumpkin-minor by using the linear-time algorithm of Courcelle~\cite{Courcelle92}.
\end{proof}

\section{A single-exponential FPT algorithm}
\label{sec:algorithm}

As mentioned in the introduction, it is proved in~\cite{FLM+11} that given an
instance $(G,k)$ of $p$-$\mathcal{F}$-\textsc{Hit} such that $\mathcal{F}$
consists of only a $c$-pumpkin for some $c \geq 1$, one can obtain in polynomial time
an equivalent instance with
 $\mathcal{O}(k^2 \log^{3/2}k)$ vertices.
This kernel leads to the following simple FPT algorithm for
$p$-$c$-\textsc{Hit}: First compute the kernel $K$ in polynomial time, and then
for every subset $S \subseteq V(K)$ of size $k$, test whether $K[V(K) \setminus
S]$ contains a $c$-pumpkin as a minor, using for instance the linear-time algorithm given by Proposition~\ref{prop:FindingTheta}. If for some $S$ we have that $K[V(K)
\setminus S]$ does not contain $c$-pumpkin as a minor, we answer \textsc{Yes};
otherwise the answer is \textsc{No}. The running time of this algorithm is
clearly bounded by $ {k^2 \log^{3/2}k \choose k}\cdot n^{\mathcal{O}(1)} =
2^{\mathcal{O}(k \log k)} \cdot n^{\mathcal{O}(1)}$.

In this section we give an algorithm for
$p$-$c$-\textsc{Pumpkin}-\textsc{Hitting} that runs in time $d^k \cdot
n^{\mathcal{O}(1)}$ for any fixed $c \geq 1$, where $d$ only depends on the
fixed constant $c$.
Towards this, we first introduce a variant of
$p$-$c$-\textsc{Pumpkin}-\textsc{Hitting}, namely
$p$-$c$-\textsc{Pumpkin}-\textsc{Disjoint Hitting}, formally defined as
follows. \vspace{.2cm}

\begin{center}
\begin{boxedminipage}{.9\textwidth}

$p$-$c$-\textsc{Pumpkin}-\textsc{Disjoint Hitting} ($p$-$c$-\textsc{Disjoint
Hit} for short)

\begin{tabular}{ r l }
\textit{~~~~Instance:} & A graph $G$ , a non-negative integer $k$, and a set $S
\subseteq V(G)$\\
& with $|S| \leq k+1$, such that $G \setminus S$ does not contain the\\
& $c$-pumpkin as a
minor. \\
\textit{Parameter:} & $k$.\\
\textit{Question:} & Does there exist $S' \subseteq V(G) \setminus S$, with $|S'| \leq k$, such that\\
~~~~~~~~~~~~~~~~~~ & $G \setminus S$ does not contain the $c$-pumpkin as a
minor?
 \\
\end{tabular}
\end{boxedminipage}

\end{center}\vspace{.2cm}

We would like to note that we will focus on solving the {\sl constructive} version of the $p$-$c$-\textsc{Disjoint
Hit} problem. That is, we will be interested in {\sl finding} such a set $S' \subseteq V(G) \setminus S$, as we will need it in our algorithm. Next we show a lemma that allows us to relate the two problems mentioned above.

\setcounter{footnote}{0}

\begin{lemma}
\label{lem:disjointtonondisjoint} If $p$-$c$-\textsc{Disjoint Hit} can be
solved in time $\eta(c)^k \cdot n^{\mathcal{O}(1)}$, then $p$-$c$-\textsc{Hit}
can be solved in time $(\eta(c)+1)^k \cdot n^{\mathcal{O}(1)}$.
\end{lemma}
\begin{proof}
Let \(\mc{A}\) be an \textsc{FPT} algorithm which solves the
$p$-$c$-\textsc{Disjoint Hit} problem in time $\eta(c)^k \cdot
n^{\mathcal{O}(1)}$. Let $(G,k)$ be an input graph for the $p$-$c$-\textsc{Hit}
problem, and let $v_1,\ldots,v_n$ be an arbitrary ordering
 of $V(G)$. Let $V_i$ and $G_i$, respectively, denote the subset $\{v_1,\ldots, v_i\}$
of vertices and the induced subgraph $G[V_i]$. We iterate over $i=1,\ldots ,n$
in the following way. At the $i$-th iteration, suppose we have a $c$-hitting
set $S_i \subseteq V_i$ of \(G_{i}\) of size at most $k$.  At the next
iteration, we set $S_{i+1}:=S_i\cup \{v_{i+1}\}$ (note that \(S_{i+1}\) is a
$c$-hitting set for $G_{i+1}$ of size at most $k+1$). If $|S_{i+1}|\leq k$, we
can safely move on to the $(i+2)$-th iteration. If $|S_{i+1}|=k+1$, we look at
every subset $S \subseteq S_{i+1}$ and check whether there is a $c$-hitting set
$W$ of size at most $k$ such that $W\cap S_{i+1}=S_{i+1}\setminus S$. To do
this, we use the FPT algorithm $\mc{A}$ for $p$-$c$-\textsc{Disjoint Hit} on
the instance $(H,S)$, with $H=G_{i+1} \setminus (S_{i+1}\setminus S)$. If
$\mc{A}$ returns a $c$-hitting set $W$ of $H$ with $|W|< |S|$, then observe
that the vertex set $(S_{i+1}\setminus S) \cup W$ is a $c$-hitting set of $G$
of size strictly smaller than $S_{i+1}$. If there is a $c$-hitting set of
$G_{i+1}$ of size strictly smaller than $S_{i+1}$, then for some $S\subseteq
S_{i+1}$, there is a small $c$-hitting set in $G_{i+1} \setminus
(S_{i+1}\setminus S)$ disjoint from $S$, and $\mc{A}$ correctly returns a
solution. If no such small $c$-hitting set exists, the algorithm returns
\textsc{No}. Let us now argue about the running time of this algorithm. The
time required to execute \(\mc{A}\) for every subset $S$ at the $i$-th
iteration is $\sum_{i=0}^{k+1}{k+1 \choose i} \cdot \eta(c)^i \cdot
n^{\mathcal{O}(1)} = (\eta(c)+1)^{k+1} \cdot n^{\mathcal{O}(1)}$. That is, we
have an algorithm for $p$-$c$-\textsc{Hit} running in time $(\eta(c)+1)^k \cdot
n^{\mathcal{O}(1)}$, as we wanted to prove.\end{proof}

Lemma~\ref{lem:disjointtonondisjoint} allows us to focus on
$p$-$c$-\textsc{Disjoint Hit}. In what follows we give an algorithm for
$p$-$c$-\textsc{Disjoint Hit} that runs in single-exponential time. In fact, we
obtain a linear kernel for $p$-$c$-\textsc{Disjoint Hit}, which clearly yields
a single-exponential algorithm.

\subsection*{Overview of the algorithm}
The algorithm for $p$-$c$-\textsc{Disjoint Hit} is based on a combination of
polynomial-time preprocessing and a protrusion-based reduction rule. Let
$(G,S,k)$ be the given instance of $p$-$c$-\textsc{Disjoint Hit} and let
$V:=V(G)$. Our main objective is to show that, after applying some simple
polynomial-time reduction rules, $\{v \in V \setminus S : N_G(v)\cap S \neq
\emptyset\}$ has cardinality $\mathcal{O}(k)$; the proof of this fact,
specially Lemma~\ref{lem:1nbr},  is the most technical part of this section.
Once we have the desired upper bound, we use a protrusion-based reduction rule
adapted from~\cite{FLM+11} to give a polynomial-time procedure that, given an
instance $(G,S,k)$ of $p$-$c$-\textsc{Disjoint Hit}, returns an equivalent
instance $(G',S,k')$ such that $G'$ has $\mathcal{O}(k)$ vertices. That is, we
obtain a linear vertex kernel for $p$-$c$-\textsc{Disjoint Hit}\footnote{This
was not the case in the conference version of this article, in which the
algorithm for $p$-$c$-\textsc{Disjoint Hit} was considerably more complicated,
involving in particular a branching procedure and a more extensive usage of the
protrusion-based reduction rule.}.
Notice that once we have a linear
vertex kernel of size $\alpha k$ for $p$-$c$-\textsc{Disjoint Hit}, we can
solve the problem in ${\alpha k \choose k } \cdot k^{\mathcal{O}(1)}$. We can
now proceed to the detailed description of the algorithm.

\subsection*{Pre-processing step} We start by defining two sets. Our first
objective is to upper bound the cardinality of these two sets by
$\mathcal{O}(k)$.

\begin{eqnarray*}
V_1& :=& \{v \in V \setminus S : |N_G(v)\cap S| = 1\}\\
V_{\geq 2} & := & \{v \in V \setminus S : |N_G(v)\cap S| \geq 2\}.
\end{eqnarray*}
We start with some simple polynomial-time reduction rules (depending on $c$)
that will be applied in the compression routine. We also prove, together with
the description of each rule, that they are valid for our problem.

\begin{itemize}
\item[{\bf R1}] Suppose that $C$ is a connected component of $G \setminus
S$ with no neighbor in $S$. Then delete $C$.

\noindent \emph{Proof of correctness}. The deletion of $C$ can be safely done,
as its vertices will never participate in a
$c$-pumpkin-model.$\hfill$$\square$\vspace{.4mm}

\item[{\bf R2}] Suppose that $C$ is a connected component of $G\setminus S$ with  exactly
one neighbor $v$ in $S$,  and such that $G[V(C) \cup \{v\}]$ is
$c$-pumpkin-free. Then delete $C$.

\noindent \emph{Proof of correctness}. In this case $C$ can be also safely
deleted, as its vertices will never participate in a minimal
$c$-pumpkin-model.$\hfill$$\square$ \vspace{.4mm}


\item[{\bf R3}]  Let $u \in S$, let $v \in V(G) \setminus S$, let $P$ be a (non-empty)
connected component of $G \setminus \{u,v\}$ entirely contained in $G \setminus
S$, and suppose that is $c$-pumpkin-free. Let $H_p$ be the graph obtained from $G[V(P) \cup
\{u,v\}]$ by adding $p$ parallel edges between $u$ and $v$, and let $p_c$ be
the smallest positive integer $p$ such that $H_p$ contains a $c$-pumpkin-minor (note that it is possible that $p_c=0$). Then
replace $P$ with $c-p_c$ parallel edges between $u$ and $v$. See
Figure~\ref{fig:R3} for some small examples for $c=4$.


%

\noindent \emph{Proof of correctness}. Note that in order to hit all $c$-pumpkins-models
intersecting $P$, there is no need to include any vertex of $P$ in the
solution, as any such vertex could be replaced with $v$, obtaining another
solution with equal or smaller size. We say that two $c$-pumpkins-models in $G$
are {\sl $P$-equivalent} if they coincide except, possibly, for vertices in
$P$. Let $G'$ be the graph obtained from $G$ by replacing $P$ with $c-p_c$
parallel edges between $u$ and $v$. By construction, $G$ and $G'$ contain
exactly the same $c$-pumpkins-models modulo the $P$-equivalence relation. As we
can assume that no vertex of $P$ is included in the solution, we conclude that
the reduction rule yields an equivalent instance. $\hfill$$\square$
\end{itemize}



\begin{figure}
\flushleft
\includegraphics[width=1.05\textwidth]{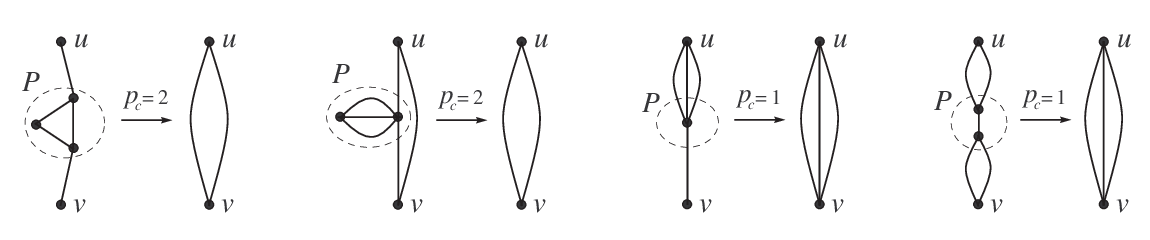}
\caption{Four examples of reduction rule~{\bf R3} for $c=4$.\label{fig:R3}}
\end{figure}


We would like to note that under the hypothesis of Rule~{\bf R3}, if in addition it holds that $G[V(P) \cup \{u,v\}]$ contains a $c$-pumpkin-minor, then we could safely delete vertex $v$ from $G$ and decrease the parameter $k$ by one. Nevertheless, it turns out that keeping $c$ parallel edges between $u$ and $v$ yields the analysis of the algorithm simpler.

We say that the instance $(G,S,k)$ is \emph{$(S,c)$-reduced} if rules~{\bf R1},
{\bf R2}, or~{\bf R3} cannot be applied anymore. Note that reduction rule~{\bf
R1} can easily be applied in polynomial time. For reduction rules~{\bf R2}
and~{\bf R3}, we have to test whether the considered graph contains a
$c$-pumpkin-minor, which can be done in linear time by Proposition~\ref{prop:FindingTheta}.


%
%

The following Lemmas~\ref{lem:2nbr} and~\ref{lem:1nbr} are key to the analysis
of our algorithm. We also need two intermediate technical results stated in
Lemmas~\ref{lem:atLeast2} and~\ref{cl:FewConnectedComponents}, which will be
used in the proof of Lemma~\ref{lem:1nbr}.


\begin{lemma}
\label{lem:2nbr} There is a function $f\colon \mathbb{N} \to \mathbb{N}$ such
that if $(G,S,k)$ is a \textsc{Yes}-instance to the $p$-$c$-\textsc{Disjoint
Hit} problem, then $|V_{\geq 2}| \leq f(c) \cdot k$.
\end{lemma}
\begin{proof}
In order to upper-bound $|V_{\geq 2}|$, we build from $G[S]$ the following
auxiliary graph $H$: we start with $H=(S,E(G[S]))$, and for each vertex $v \in
V_{\geq 2}$ with neighbors $u_1,\ldots,u_{\ell}$, $\ell \geq 2$, we add to $H$
an edge $e_v$ between two arbitrary neighbors $u_1,u_2$ of $v$. Note that $H
\preceq_m G$, and that for each vertex $v \in V_{\geq 2}$, $H \setminus e_v
\preceq_m G \setminus v$. We now argue that $|E(H)|$ is linearly bounded by
$k$, which implies the desired result as by construction $|V_{\geq 2}| \leq
|E(H)|$. If $G$ is a \textsc{Yes}-instance, there must be a set $S' \subseteq V
\setminus S$, $|S'| \leq k$, such that $G \setminus S'$ is $c$-pumpkin-free. By
construction of $H$, the removal of each vertex $v \in S' \cap V_{\geq 2}$
corresponds to the removal of the edge $e_v \in E(H)$. Let $H' = H \setminus
\{e_v : v \in S' \cap V_{\geq 2}\}$, and note that $H' \preceq_m G \setminus
S'$, so $H'$ is $c$-pumpkin-free. Therefore, by
Corollary~\ref{cor:MultiExcludeTheta} it follows that $|E(H')| \leq (c-1)\cdot
(2c-1)\cdot |V(H')| \leq (c-1)\cdot (2c-1)\cdot (k+1)$. As $|E(H)| \leq
|E(H')|+k$, we conclude that $|V_{\geq 2}| \leq |E(H)| \leq (c-1)\cdot
(2c-1)\cdot (k+1) +k$.\end{proof}

\begin{lemma}
\label{lem:atLeast2} There is a function $g\colon \mathbb{N} \to \mathbb{N}$
such that if $(G,S,k)$ is a \textsc{Yes}-instance to the
$p$-$c$-\textsc{Disjoint Hit} problem and $\mathcal{C}$ is a collection of
disjoint connected subgraphs of $G \setminus S$ such that each subgraph has at
least two distinct neighbors in $S$, then $|\mathcal{C}| \leq g(c) \cdot k$.
\end{lemma}
\begin{proof}The proof is very similar to the proof of Lemma~\ref{lem:2nbr}. In order
to upper-bound $|\mathcal{C}|$, we build from $G[S]$ the following auxiliary
graph $H$: we start with $H=(S,E(G[S]))$, and for each subgraph $C \in
\mathcal{C}$ with neighbors $u_1,\ldots,u_{\ell}$, $\ell \geq 2$, we add to $H$
an edge $e_C$ between two arbitrary neighbors $u_1,u_2$ of $C$. Note that $H
\preceq_m G$, and that for each subgraph $C \in \mathcal{C}$, $H \setminus e_C
\preceq_m G \setminus C$. We now argue that $|E(H)|$ is linearly bounded by
$k$, which implies the desired result as by construction $|\mathcal{C}| \leq
|E(H)|$. If $G$ is a \textsc{Yes}-instance, there must be a set $S' \subseteq V
\setminus S$, $|S'| \leq k$, such that $G \setminus S'$ is $c$-pumpkin-free. By
construction of $H$, the removal of a vertex $v$ in a subgraph $C \in
\mathcal{C}$ corresponds to the removal of at most one edge $e_C \in E(H)$ (as
maybe the edge $e_C$ can still be {\sl simulated} after the removal of $v$).
Let $H'$ be the subgraph obtained from $H$ after the removal of those edges,
and note that $H' \preceq_m G \setminus S'$, so $H'$ is $c$-pumpkin-free.
Therefore, by Corollary~\ref{cor:MultiExcludeTheta} it follows that $|E(H')|
\leq (c-1)\cdot (2c-1)\cdot |V(H')| \leq (c-1)\cdot (2c-1)\cdot (k+1)$. As
$|E(H)| \leq |E(H')|+k$, we conclude that $|\mathcal{C}| \leq |E(H)| \leq
(c-1)\cdot (2c-1)\cdot (k+1) +k$.
\end{proof}

\begin{lemma}
\label{cl:FewConnectedComponents} In an $(S,c)$-reduced \textsc{Yes}-instance
$(G,S,k)$ to the
$p$-$c$-\textsc{Disjoint Hit} problem, the number of connected components of $G\setminus S$ is
$\mathcal{O}(k)$.
\end{lemma}
\begin{proof} Note that by reduction rules~{\bf R1} and~{\bf R2}, we can assume that
each connected component $C$ of $G\setminus S$ has some neighbor in $S$, and
that if $C$ has exactly one neighbor $v$ in $S$, then $G[V(C) \cup \{v\}]$ has
a $c$-pumpkin. On the one hand, the number of components $C$ that have exactly
one neighbor $v$ in $S$ and such that $G[V(C) \cup \{v\}]$ contains the
$c$-pumpkin as a minor is at most $k$, as any solution needs to contain at
least one vertex from each such connected component. On the other hand, the
number of components that have at least two neighbors in $S$ is
$\mathcal{O}(k)$ by Lemma~\ref{lem:atLeast2}.\end{proof}

Now we prove our key structural lemma.

\begin{lemma}
\label{lem:1nbr} There is a function $h\colon \mathbb{N} \to \mathbb{N}$ such
that if $(G,S,k)$ is an $(S,c)$-reduced \textsc{Yes}-instance to the
$p$-$c$-\textsc{Disjoint Hit} problem, then  $|V_{1}| \leq h(c) \cdot k$.
\end{lemma}
\begin{proof}
For simplicity we call the vertices in $V_{1}$ {\sl white}. We proceed to find
a packing of disjoint connected subgraphs $\mathcal{P}=\{B_1,\ldots,B_{\ell}\}$
of $G \setminus S$ containing all white vertices except for $\mathcal{O}(k)$ of
them. This will help us in bounding $|V_{1}|$. We call the subgraphs in
$\mathcal{P}$ \emph{blocks}. For a graph $H \subseteq G\setminus S$, let $w(H)$
be the number of white vertices in $H$. The idea is to obtain, as far as
possible, blocks $B$ with $c \leq w(B) \leq c^3$; we call these blocks
\emph{suitable}, and the other blocks are called \emph{unsuitable}. If at some
point we cannot refine the packing anymore in order to obtain suitable blocks,
we will argue about its structural properties,
which will allow us to bound the number of white vertices.


We start with $\mathcal{P}$ containing all the connected components $C$ of $G
\setminus S$ such that $w(C)> c^3$, and we recursively try to refine the
current packing. By Lemma~\ref{cl:FewConnectedComponents}, we know that the
number of connected components is $\mathcal{O}(k)$, and hence the number of
white vertices that are not included in $\mathcal{P}$ is
$\mathcal{O}(c^3k)=\mathcal{O}(k)$.


More precisely, for each block $B$ with $w(B) > c^3$, we build a spanning tree
$T$ of $B$, and we orient each edge $e \in E(T)$ towards the components of $T
\setminus \{e\}$ containing at least $c$ white vertices. Note that, as $w(B) >
c^3$, each edge gets at least one orientation, and that edges may be oriented
in both directions. If some edge $e \in E(T)$ is oriented in both directions,
we replace in $\mathcal{P}$ block $B$ with the subgraphs induced by the
vertices in each of the two subtrees. We stop this recursive procedure whenever
we cannot find more suitable blocks using this orientation trick. Let
$\mathcal{P}$ be the current packing.

Now let $B$ be an unsuitable block in $\mathcal{P}$, that is, $w(B)>c^3$ and no
edge of its spanning tree $T$ is oriented in both directions. This implies that
there exists a vertex $v \in V(T)$ with all its incident edges pointing towards
it. We call such a vertex $v$ a \emph{sink}. Let $T_1,\ldots,T_p$ be the
connected components of $T \setminus \{v\}$. Note that as $v$ is a sink,
$w(T_i) < c$ for $1 \leq i \leq p$, using the fact that $w(B)
> c^3$  we conclude that $p \geq c^2$.
Now let $P_1,\ldots,P_{\ell}$ be the connected components of $G[V(T_1)\cup
\cdots \cup V(T_p)]= G[V(B) \setminus \{v\}]$, and note that $\ell \leq p$. We
call these subgraphs $P_i$ the \emph{pieces} of the unsuitable block $B$. For
each unsuitable block, we delete the pieces with no white vertex. This
completes the construction of $\mathcal{P}$. The next claim bounds the number
of white vertices in each piece of an unsuitable block in $\mathcal{P}$.

\begin{claimN}
\label{cl:pizza} Each of the pieces of an unsuitable block contains less than
$c^2$ white vertices.
\end{claimN}
\begin{proof}
Assume for contradiction that there exists a piece $P$ of an unsuitable block
with $w(P) \geq c^2$, and let $v$ be the sink of the unsuitable block obtained
from tree $T$. By construction, $V(P)$ is the union of the vertices in some of
the trees in $T \setminus \{v\}$; let without loss of generality these trees be
$T_1,\ldots,T_{q}$. As $w(P) \geq c^2$ and $w(T_i) < c$ for $1 \leq i \leq q$,
it follows that $q \geq c$. As $v$ has at least one neighbor in each of the
trees $T_i$, $1 \leq i \leq q$, and $P$ is a connected subgraph of $G$, we can
obtain a $c$-pumpkin-model $\{A,B\}$ in $G \setminus S$ by setting $A:=\{v\}$
and $B:=V(P)$, contradicting the fact that $G \setminus S$ is $c$-pumpkin-free.
See Figure~\ref{fig:packing}(a) for an example for $c=3$.\end{proof}

\begin{figure}
\flushleft
\includegraphics[width=1.00\textwidth]{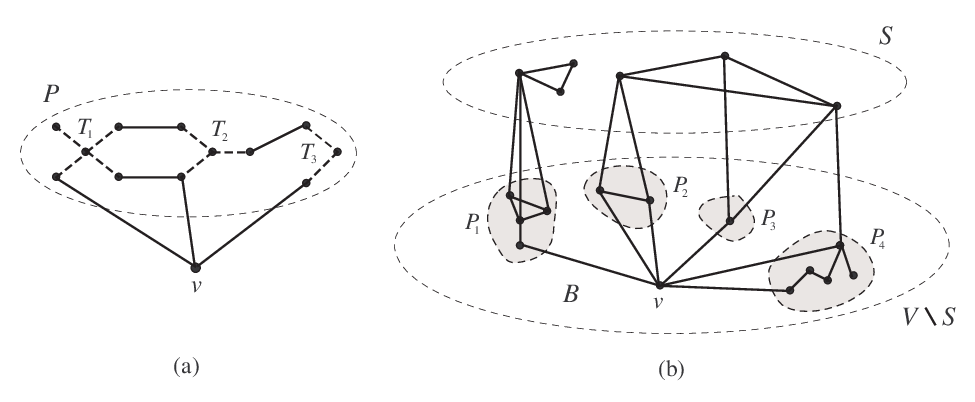}
\caption{(a) Example for $c=3$ of the contradiction in the proof of
Claim~\ref{cl:pizza}. The trees $T_1$, $T_2$, and $T_3$ are defined by the
dashed edges. (b) Example for $c=3$ of an unsuitable block $B$ in the packing
$\mathcal{P}$ built in the proof of Lemma~\ref{lem:1nbr}\label{fig:packing}.}
\end{figure}


Hence in the packing $\mathcal{P}$ we are left with a set of suitable blocks
with at most $c^3$ white vertices each, and a set of unsuitable blocks, each
one broken up into pieces linked by a sink in a star-like structure. By
Claim~\ref{cl:pizza}, each piece of the remaining unsuitable blocks contains at
most $c^2$ white vertices. See Figure~\ref{fig:packing}(b) for an example of an
unsuitable block $B$ for $c=3$.


Now we need two claims about the properties of the constructed packing.

\begin{claimN}
\label{cl:suitable} In the packing $\mathcal{P}$ constructed above, the number
of suitable or unsuitable blocks is $\mathcal{O}(k)$.
\end{claimN}
\begin{proof}
We first bound the number of suitable blocks, and for this we distinguish
between two types of suitable blocks.
\begin{itemize}
\item[$\bullet$] Type 1: blocks that
have exactly one neighbor in $S$. For each such block we need to include a
vertex of it in the $c$-hitting set (as each suitable block contains at least
$c$ white vertices), so their number is at most $k$.\vspace{.4mm}
\item[$\bullet$] Type 2: blocks that have at least two distinct neighbors in $S$.
The number of such blocks is $\mathcal{O}(k)$ by Lemma~\ref{lem:atLeast2}.
\end{itemize}
The proof for unsuitable blocks is similar. Namely, we distinguishing between the same
two types of blocks, and use the fact that each unsuitable block contains at
least $c^3\geq c$ white vertices. This concludes the proof.\end{proof}



\begin{claimN}
\label{cl:numberPieces} In an $(S,c)$-{\sl reduced} \textsc{Yes}-instance, the
total number of pieces in the packing $\mathcal{P}$ constructed above is
$\mathcal{O}(k)$.
\end{claimN}
\begin{proof}
We distinguish between three types of pieces.
\begin{itemize}
\item[$\bullet$] Type 1: pieces that have at least two distinct neighbors
in $S$ (see piece $P_3$ in Figure~\ref{fig:packing}(b) for $c=3$). The number
of pieces of this type is $\mathcal{O}(k)$ by
Lemma~\ref{lem:atLeast2}.\vspace{.4mm}
\item[$\bullet$] Type 2: pieces that are not of Type 1 and that have at least one neighbor
in some suitable block or in another unsuitable block (note that by
construction a piece cannot have any neighbor in other pieces of the same
unsuitable block). We construct an auxiliary graph $H$ as follows: we start
with the packing $\mathcal{P}$, and we add all the edges in $G \setminus S$
between vertices in different blocks of $\mathcal{P}$ (suitable or unsuitable).
Then we contract each block to a single vertex, and let $H$ be the resulting
graph. By Claim~\ref{cl:suitable}, $|V(H)|=\mathcal{O}(k)$. As $H \preceq_m G
\setminus S$ and $G \setminus S$ is $c$-pumpkin-free, by
Corollary~\ref{cor:MultiExcludeTheta} we have that $|E(H)|=\mathcal{O}(k)$.
Since each piece of Type~2 is incident with at least one edge of $H$ after
uncontracting the blocks, it follows that the number of pieces of Type~2 is at
most $2|E(H)|=\mathcal{O}(k)$.\vspace{.4mm}
\item[$\bullet$] Type 3: the remaining pieces. That is, these are pieces
$P$ that see exactly one vertex $u$ in $S$, and that are connected to the rest
of $G \setminus S$ only through the corresponding sink $v$. In other words,
such a piece $P$ is a connected component of $G \setminus \{u,v\}$. We
distinguish two subcases.
\begin{itemize}
\item[$\circ$] If $G[V(P) \cup \{u\}]$ contains a $c$-pumpkin-minor, then any
$c$-hitting set needs to contain at least one vertex in $P$ (see piece $P_1$ in
Figure~\ref{fig:packing}(b) for $c=3$). We conclude that the number of pieces
of this subtype is at most $k$.
\item[$\circ$] Otherwise, all the conditions of reduction rule~{\bf R3}
are fulfilled, and we conclude that such a piece cannot exist in an
$(S,c)$-{\sl reduced} instance.
\end{itemize}
\end{itemize}\vspace{-.4cm}
\end{proof}



To conclude, recall that the constructed packing $\mathcal{P}$ contains all but
$\mathcal{O}(k)$ white vertices, either in suitable blocks or in pieces of
unsuitable blocks. As by construction each suitable block has at most $c^3$
white vertices and by Claim~\ref{cl:suitable} the number of such blocks is
$\mathcal{O}(k)$, it follows that the number of white vertices contained in
suitable blocks is $\mathcal{O}(k)$. Similarly, by Claim~\ref{cl:pizza} each
piece contains at most $c^2$ white vertices, and the total number of pieces is
$\mathcal{O}(k)$ by Claim~\ref{cl:numberPieces}, so the number of white
vertices contained in pieces of unsuitable blocks is also
$\mathcal{O}(k)$.\end{proof}

\subsection*{Linear kernel}
We now proceed to describe a procedure, called \emph{protrusion rule}, that
bounds the size of our graph when $|V_1\cup V_{\geq 2}|=\mathcal{O}(k)$. We
first need some definitions.

%



Many graph optimization problems can be expressed as finding an optimal number
of vertices or edges satisfying a property expressible in Monadic Second Order
(MSO) logic. A parameterized graph problem $\Pi \subseteq \Sigma^* \times
\mathbb{N}$ is given with a graph $G$ and an integer $k$ as an input. When the
goal is to decide whether there exists a subset $W$ of at most $k$ vertices for
which an MSO-expressible property $P_{\Pi}(G,W)$ holds, we say that $\Pi$ is a
\mMSO ~graph problem. One can easily check that the $p$-$c$-\textsc{Hit}
problem is \mMSO. In the (parameterized) \emph{disjoint version} $\Pi^d$ of a
\mMSO ~problem $\Pi$, we are given a triple $(G,S,k)$, where $G$ is a graph,
$S$ a subset of $V(G)$ and $k$ the parameter, and we seek for a solution set
$W$ which is disjoint from $S$, and whose size is at most $k$.


Given $R\subseteq V(G)$, we define $\partial_G(R)$ as the set of vertices in
$R$ that have a neighbor in $V(G)\setminus R$. Thus the neighborhood of $R$ is
$N_G(R) = \partial_G(V(G) \setminus R)$. We say that a set $X \subseteq V(G)$
is an {\em $r$-protrusion} of $G$ if $\tw(G[X])\leq r$ and $|\partial_G(X)|
\leq r$.

An important concept when applying protrusion-based reduction rules is
\emph{strong monotonicity} of a problem, which we do not define here (see for
instance~\cite{BodlaenderFLPST09}). What we will need is the following fact,
which can be found in~\cite[proof of Lemma 13]{BodlaenderFLPST09}: if
$\mathcal{F}$ is a finite set of connected planar graphs, then the
$p$-$\mathcal{F}$-\textsc{Hit} problem is strongly monotone. As the $c$-pumpkin
is a connected planar graph for any $c \geq 1$, we immediately have the
following lemma.

\begin{lemma}
\label{lem:strongMon} The $p$-$c$-\textsc{Hit} problem is strongly monotone.
\end{lemma}

The following lemma is key to our protrusion-based reduction rule. Its proof
follows basically from the framework introduced in~\cite{BodlaenderFLPST09},
although some details need to be modified when dealing with the disjoint
version of a problem, as it is our case. A self-contained proof for the
specific case of disjoint problems can be found in~\cite[Lemma 12]{KPP11}. The
general statement deals with the disjoint version of a general (parameterized)
strongly monotone \mMSO ~problem. As the $p$-$c$-\textsc{Hit} problem is \mMSO,
and it is strongly monotone by Lemma~\ref{lem:strongMon}, we only state the
lemma for the specific case of our problem.

\begin{lemma}[Kim \emph{et al}.~\cite{KPP11}]
\label{safe:protrusion} Let $\Pi^d$ be the disjoint version of
$p$-$c$-\textsc{Hit}. There exists a computable function
$\gamma:\mathbb{N}\rightarrow \mathbb{N}$ and an algorithm that given:
\begin{itemize}
\item an instance $(G,S,k)$ of $\Pi^d$ such that $G \setminus S$ is $c$-pumpkin-minor-free, and
\item a $t$-protrusion $X$ of \(G\) such that $|X|>\gamma(2t+1)$ and
  $X\cap S=\emptyset$,
\end{itemize}
in time $\mathcal{O}(|X|)$ outputs an instance $(G',S,k')$ such that
$|V(G')|<|V(G)|$, $k'\leq k$,  $(G',S,k')\in\Pi^d$ if and only if
$(G,S,k)\in\Pi^d$, and $G' \setminus S$ is $c$-pumpkin-minor-free.
\end{lemma}

We are now ready the state the protrusion rule. It follows as a corollary of
Lemma~\ref{safe:protrusion} that the following reduction rule for
$p$-$c$-\textsc{Disjoint Hit} is safe.



\begin{itemize}
\item[{\bf P}] Let $(G,S,k)$ be an instance of $p$-$c$-\textsc{Disjoint Hit} and let
$\gamma:\mathbb{N}\rightarrow \mathbb{N}$ be the computable function given by
Lemma~\ref{safe:protrusion}. Let $X$ be a $4c$-protrusion of $G$ with
$|X|>\gamma(8c+1)$ and such that $X\cap S=\emptyset$. Then use the algorithm
given by Lemma~\ref{safe:protrusion} to compute in time $\mathcal{O}(|X|)$ an
equivalent instance $(G',S,k')$ such that $G[S]$ and $G'[S]$ are isomorphic,
$G' \setminus S$ is $c$-pumpkin-minor-free, $|V(G')|<|V(G)|$, and $k'\leq k$.
\end{itemize}


Before describing how to obtain the linear kernel for $p$-$c$-\textsc{Disjoint
Hit}, we need the following lemma, corresponding to~\cite[Lemma~$6$]{FLM+11}.

\begin{lemma}[Fomin \emph{et al}.~\cite{FLM+11}]
\label{lem:FindProtrusions} There is a linear-time algorithm that given an
$n$-vertex graph $G$ and a set $S \subseteq V(G)$ such that $\tw(G \setminus S)
\leq d$, outputs a $2(d + 1)$-protrusion in $G \setminus S$  of size at least
$\frac{n-|S|}{4|N(S)|+1}$, where $N(S)$ is the set of vertices in $V(G)
\setminus S$ with at least one neighbor in $S$. Here $d$ is some constant.
\end{lemma}

The proof of the next lemma is similar to the proof
of~\cite[Theorem~$1$]{FLM+11}.



\begin{lemma}
\label{lemma:linearkernel} If  $|V_1\cup V_{\geq 2}|=\mathcal{O}(k)$, then
$p$-$c$-\textsc{Disjoint Hit} has a kernel with $\mathcal{O}(k)$ vertices.
\end{lemma}
\begin{proof}
Let for this proof $V_{\star}=V_1\cup V_{\geq 2}$, and let $(G,S,k)$ be an
instance of $p$-$c$-\textsc{Disjoint Hit} such that
$|V_{\star}|=\mathcal{O}(k)$. As by Theorem~\ref{cor:ExcludeTheta} we have that
$\tw(G \setminus S)\leq 2c-1$, we can apply Lemma~\ref{lem:FindProtrusions} and
obtain in linear time a $2((2c-1)+1)=4c$-protrusion $Y$ of size at least
$\frac{|V(G)|-|S|}{4|V_{\star}|+1}$ in $G\setminus S$. Let $\gamma : \mathbb{N}
\rightarrow \mathbb{N}$  be the function defined in
Lemma~\ref{safe:protrusion}. If $\frac{|V(G)|-|S|}{4|V_{\star}|+1} >
\gamma(8c+1)$, then using protrusion rule {\bf P} we obtain in time
$\mathcal{O}(|Y|)$ an instance $(G',S,k')$ such that $G[S]$ and $G'[S]$ are
isomorphic, $G' \setminus S$ is $c$-pumpkin-minor-free, $|V(G')|<|V(G)|$,
$k'\leq k$, and such that $(G',S,k')$ is a \textsc{Yes}-instance of
$p$-$c$-\textsc{Disjoint Hit} if and only if $(G,S,k) $ is a
\textsc{Yes}-instance of $p$-$c$-\textsc{Disjoint Hit}. We continue applying
Lemma~\ref{lem:FindProtrusions} and protrusion rule {\bf P} to the newly
obtained instance as far as there is a $4c$-protrusion of size strictly greater
than $\gamma(8c+1)$. We would like to note that in the whole process we never
delete either the vertices of $S$ or $V_{\star}$.


Let $(G^*,S,k^*)$ be the reduced instance obtained after this procedure. It
follows that there is no $4c$-protrusion of size greater than $\gamma(8c+1)$ in
$G^*\setminus S$, so protrusion rule {\bf P} no longer applies. Note that
$k^*\leq k$. We claim that the number of vertices in this graph $G^*$  is
bounded by $\mathcal{O}(k)$. Indeed, since we cannot apply protrusion rule {\bf
P}, we have that $\frac{|V(G^*)|-|S|}{4|V_{\star}|+1}\leq \gamma(8c+1)$.
Therefore, we have that $ |V(G^*)|  \leq \gamma(8c+1)
\cdot(4|V_{\star}|+1)+|S|$. Since $|S|= \mathcal{O}(k)$ and by hypothesis
$|V_{\star}|=\mathcal{O}(k)$, we have that $ |V(G^*)|= \mathcal{O}(k)$. This
completes the proof.\end{proof}

Lemma~\ref{lemma:linearkernel} clearly implies that, if $|V_1\cup V_{\geq
2}|=\mathcal{O}(k)$, then $p$-$c$-\textsc{Disjoint Hit} can be solved in time
$2^{\mathcal{O}(k)} \cdot n^{\mathcal{O}(1)}$. Nevertheless, the above proof
only shows that the {\sl decision version} of $p$-$c$-\textsc{Disjoint Hit} can
be solved in single-exponential time, as we have applied reduction rules that
may modify the instance. But in the iterative compression routine (see proof of
Lemma~\ref{lem:disjointtonondisjoint}), we need to be able to obtain an {\sl
explicit solution} $S'\subseteq V(G)\setminus S$ of $p$-$c$-\textsc{Disjoint
Hit} in the original instance, with $|S'|=k$, if it exists.

We can get this explicit solution (if it exists) by repeatedly applying the
single-exponential algorithm for the decision version as follows. Suppose that
$G$ is a \textsc{Yes}-instance of $p$-$c$-\textsc{Disjoint Hit}, and let an
ordering of the vertices of $V(G)\setminus S$ be $u_1,u_2,\ldots,u_q$. Set
$i:=1$ and  $U:=\emptyset$. Repeat the following two steps for $i=1,\ldots,q$.

\begin{enumerate}
\item Check whether $G\setminus U$ is $c$-pumpkin-free in linear time, using Proposition~\ref{prop:FindingTheta}. If it is the
case, then return $U$ as the desired solution. Otherwise, go to the next step.
\item Using the
single-exponential algorithm for the decision version of
$p$-$c$-\textsc{Disjoint Hit}, check whether $G\setminus (\{u_i\}\cup U)$
contains a solution $S^* \subseteq V(G)\setminus (S \cup U \cup \{u_i\})$ of
size $k-(|U|+1)$. If it is the case, then set $U:=U\cup \{u_i\}$. Set $i:=i+1$.
\end{enumerate}
This concludes the description of the algorithm to obtain the desired explicit
solution $U$ in the compression step.

\subsection*{Final algorithm}
Finally we combine everything to obtain the following result.
\begin{theorem}
\label{prop:runningTime} For any fixed $c \geq 1$, the
$p$-$c$-\textsc{Pumpkin}-\textsc{Hitting} problem can be solved in time
$2^{\mathcal{O}(k)} \cdot n^{\mathcal{O}(1)}$.
\end{theorem}
\begin{proof}
To obtain the desired result, by Lemma~\ref{lem:disjointtonondisjoint} and the
procedure described after the proof of Lemma~\ref{lemma:linearkernel}, it is
sufficient to obtain a single-exponential algorithm for
$p$-$c$-\textsc{Disjoint Hit}. To this end, after applying reduction rules {\bf
R1}, {\bf R2}, and {\bf R3} in polynomial time, by Lemmas~\ref{lem:2nbr}
and~\ref{lem:1nbr} we have that $|V_1 \cup V_{\geq 2}|= \mathcal{O}(k)$. Thus,
using Lemma~\ref{lemma:linearkernel} we get, also in polynomial time, an
equivalent instance $(G^*,S,k^*)$ with $\mathcal{O}(k)$ vertices, and hence the
problem can be solved by enumerating all subsets of size at most $k^*$ of
$G^*\setminus S$. This completes the proof.\end{proof}

\subsection*{Running time analysis.} To conclude this section, we provide a running time analysis of the algorithm given by Theorem~\ref{prop:runningTime} above. We would like to note that we did not focus at all on optimizing the hidden constant in the term $2^{\mathcal{O}(k)}$, so we will just concentrate on the term $n^{\mathcal{O}(1)}$. From the proof of Lemma~\ref{lem:disjointtonondisjoint} it follows that if $p$-$c$-\textsc{Disjoint Hit} can be
solved in time $a^k \cdot n^b$ for two constants $a,b$, then $p$-$c$-\textsc{Pumpkin}-\textsc{Hitting}
can be solved in time $(a+1)^k \cdot n^{b+1}$. Let us now focus on $p$-$c$-\textsc{Disjoint Hit}. First note that reduction rules {\bf R1}, {\bf R2} and {\bf R3} can be applied in linear time. Indeed, the connected components of $G \setminus S$ can be listed by successively performing BFS in time $\mathcal{O}(|V(G \setminus S)| + |E(G \setminus S)|)$, which equals
$\mathcal{O}(|V(G \setminus S)|)$ as the graph $G \setminus S$ has bounded treewidth by Theorem~\ref{cor:ExcludeTheta}. By Proposition~\ref{prop:FindingTheta}, testing for the existence of a $c$-pumpkin-minor can also be performed in linear time. As for protrusion rule {\bf P}, it can also be performed in linear time by Lemmas~\ref{safe:protrusion} and~\ref{lem:FindProtrusions}. As each of these reduction rules is applied at most $\mathcal{O}(k \cdot n)$ times, and as once we have a linear kernel the problem can be solved exhaustively in time $2^{\mathcal{O}(k)}$, we conclude that $p$-$c$-\textsc{Disjoint Hit} can be
solved in time $2^{\mathcal{O}(k)} \cdot n^2$, and therefore the algorithm given by Theorem~\ref{prop:runningTime} solves the $p$-$c$-\textsc{Pumpkin}-\textsc{Hitting} problem in time $2^{\mathcal{O}(k)} \cdot n^3$. We feel that there is room for improvement in this running time, as it was not our main objective to optimize it.


\section{An approximation algorithm for hitting and packing pumpkins}
\label{sec:combinatorial}

In this section we show that every $n$-vertex graph $G$ either contains a small
$c$-pumpkin-model or has a structure that can be reduced, giving a smaller
equivalent instance for both the \textsc{Minimum}
$c$-\textsc{Pumpkin}-\textsc{Hitting} and the \textsc{Maximum}
$c$-\textsc{Pumpkin}-\textsc{Packing} problems. Here by a ``small''
$c$-pumpkin-model, we mean a model of size at most $f(c)\cdot \log n$ for some
function $f$ independent of $n$. We finally use this result to derive a
$\mathcal{O}(\log n)$-approximation algorithm for both problems.

This section is organized as follows. We first describe in
Section~\ref{sec:reduction_rules} our reduction rules and prove their validity
for both hitting and packing problems (see Lemma~\ref{lem:both}). The existence
of small $c$-pumpkin-models in $c$-reduced graphs is provided in
Section~\ref{sec:smallPumpkins} (see Lemma~\ref{lem:smallmodelsreducedgraphs});
its proof strongly relies on a graph structure that we call \emph{hedgehog},
which we study in Section~\ref{sec:hedgehogs}. We finally focus in
Section~\ref{sec:1stAlgo} on the algorithmic consequences of our results (see
Theorem~\ref{thm:logn_packing_covering}).

\subsection{Reduction rules}
\label{sec:reduction_rules}

We describe two reduction rules for hitting/packing $c$-pumpkin-models, which
given an input graph $G$ satisfying some specific conditions, produce a graph
$H$ with less vertices than $G$ and satisfying $\tau_{c}(G)= \tau_{c}(H)$ and
$\nu_{c}(G)= \nu_{c}(H)$. Moreover, these operations are defined in such a way
that, for both problems, an optimal (resp.\ approximate) solution for $G$ can
be retrieved in polynomial time from an optimal (resp.\ approximate) solution
for $H$.

Given a graph $G$ and two distinct vertices $u, v$ of $G$, we write $G +_{k}
uv$ for the graph obtained from $G$ by adding $k$ parallel edges linking $u$ to
$v$. A {\em $c$-outgrowth} of a graph $G$ is a triple $(K, u, v)$ such that

\begin{itemize}
\item[(i)] $u,v$ are two distinct vertices of $G$;\vspace{.4mm}
\item[(ii)] $K$ is a connected component of $G \setminus \{u, v\}$ with $|K|
\geq 1$; \vspace{.4mm}
\item[(iii)] $u$ and $v$ both have at least one neighbor in $K$ in
the graph $G$; and
\item[(iv)] the graph $\Gamma(K, u, v)$ obtained from $G[V(K) \cup \{u,v\}]$ by removing
all the edges between $u$ and $v$ has no $c$-pumpkin-minor.
\end{itemize}

Given a $c$-outgrowth $(K, u, v)$ we let $\gamma(K, u, v) := c - k$, where $k$
is the smallest integer such that $\Gamma(K, u, v) +_{k} uv$ has a $c$-pumpkin
minor. Observe that, when adding $k$ parallel edges between $u$ and $v$ to
$\Gamma(K, u, v)$, there are two distinct ``types'' of $c$-pumpkin-models $\{A,
B\}$ that could appear: Exchanging $A$ and $B$ if necessary, we either have $u
\in A$ and $v\in B$ (first type), or $u, v\in A$ (second type). (Possibly both
types of models coexist.) Note that we always have $\gamma(K, u, v) = c - 1$
when $\Gamma(K, u, v) +_{k} uv$ contains a $c$-pumpkin-model of the second
type. See Figure~\ref{fig:outgrowth} for an illustration.

\begin{figure}
\centering
\includegraphics[width=0.63\textwidth]{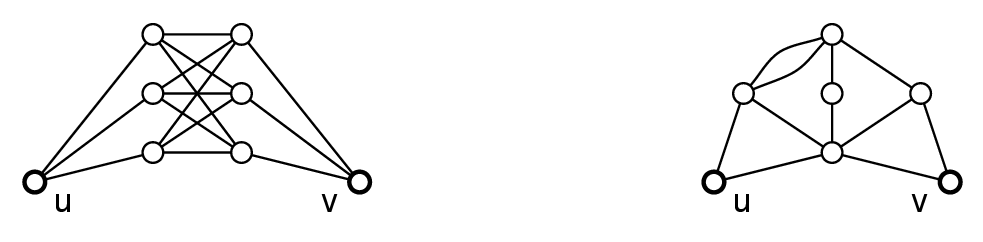}
\caption{The graph $\Gamma(K, u, v)$ of two $c$-outgrowths $(K, u, v)$ for
$c=11$ (left) and $c=7$ (right), respectively. We have $\gamma(K, u, v) = 9$
for the left one, and $\gamma(K, u, v) = 6$ for the right one.
\label{fig:outgrowth}} \vspace{-.15cm}
\end{figure}

Now that we are equipped with these definitions and notations, we may describe
the two reduction rules, which depend on the value of the positive integer $c$.

\begin{enumerate}
\item[{\bf Z1}]
Suppose $v$ is a vertex of $G$ such that no block
of $G$ containing $v$ has a $c$-pumpkin-minor. Then define $H$ as the graph
obtained from $G$ by removing $v$.\vspace{.4mm}
\item[{\bf Z2}]
Suppose that $(K, u, v)$ is a $c$-outgrowth of $G$.
Then define $H$ as the graph obtained from $G \setminus V(K)$ by adding
$\gamma(K, u, v)$ parallel edges between $u$ and $v$.
\end{enumerate}

\begin{figure}[t]
\centering
\includegraphics[width=0.65\textwidth]{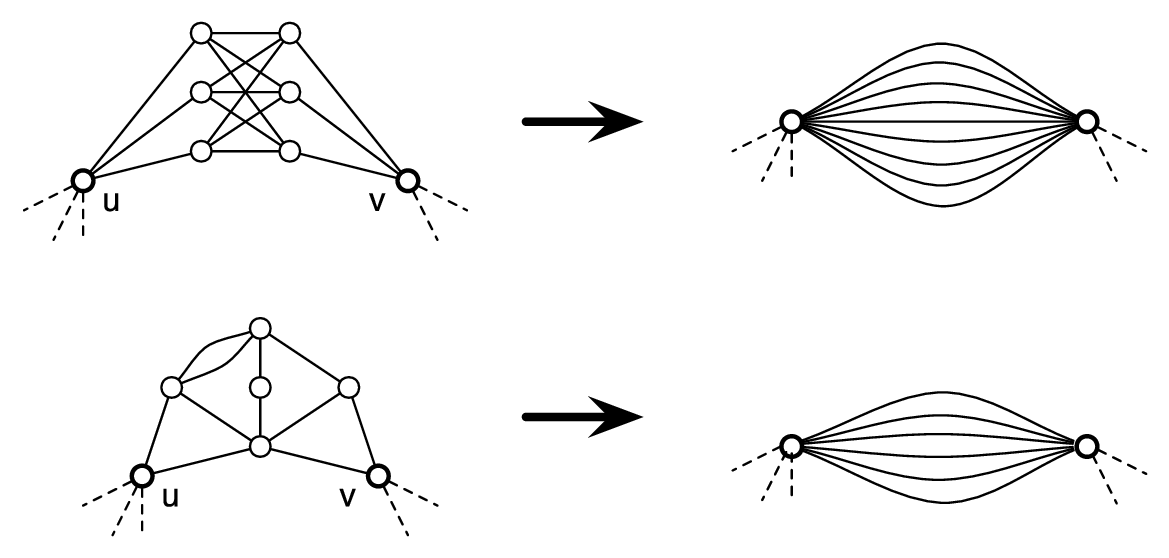}
\caption{Illustration of reduction rule {\bf Z2} on the two outgrowths from
Figure~\ref{fig:outgrowth}.\label{fig:Z2}}
\end{figure}

See Figure~\ref{fig:Z2} for an illustration of {\bf Z2}.

We note that testing for the existence of a $c$-pumpkin-minor in a graph $G$
can be done in polynomial time when $c$ is fixed by Proposition~\ref{prop:FindingTheta}. Moreover, if there is one, an
explicit $c$-pumpkin-model can be computed. This
follows from classical results of Robertson and Seymour~\cite{RS95}, and
will be used implicitly in the
subsequent proofs. Note that, in particular, testing whether a vertex $v$ is in
a block containing a $c$-pumpkin-minor can be done in polynomial time.
Similarly, testing whether a triple $(K, u, v)$ with $K$ a component of $G -
\{u,v\}$ is a $c$-outgrowth can be done in polynomial time, and the parameter
$\gamma(K, u, v)$ can be computed in polynomial time as well. Therefore, we can
check in polynomial time if {\bf Z1} or {\bf Z2} can be applied to a given
graph, and each of these two reduction rules can be realized in polynomial
time.

A graph $G$ is said to be {\em $c$-reduced} if neither {\bf Z1} nor {\bf Z2}
can be applied to $G$. The next lemma shows the validity of these reduction
rules.

\begin{lemma}
\label{lem:both} Let $c$ be a fixed positive integer. Suppose that $H$ results
from the application of {\bf Z1} or {\bf Z2} on a graph $G$. Then
\begin{itemize}
\item[(a)] $\tau_{c}(G) = \tau_{c}(H)$ and moreover,
given a $c$-hitting set $X'$ of $H$, one can compute in polynomial time a
$c$-hitting set $X$ of $G$ with $|X| \leq |X'|$.\vspace{.4mm}
\item[(b)] $\nu_{c}(G) = \nu_{c}(H)$ and moreover, given a $c$-packing
$\mathcal{M}'$ of $H$, one can compute in polynomial time a $c$-packing
$\mathcal{M}$ of $G$ with $|\mathcal{M}| = |\mathcal{M}'|$. \end{itemize}
\end{lemma}

In order to prove Lemma~\ref{lem:both}, we first need to introduce a few
technical lemmas; the validity of the reduction rules is shown in
Lemmas~\ref{lem:valid_covering} and~\ref{lem:valid_packing} at the end of this
section, which correspond to Lemma~\ref{lem:both}(a) and
Lemma~\ref{lem:both}(b), respectively.

\begin{lemma}
\label{lem:intermediate1} Let $c$ be a fixed positive integer. Suppose $H$ is
obtained by applying rule {\bf Z2} on a $c$-outgrowth $(K, u, v)$ of a graph
$G$. Let $X$ be an arbitrary subset of vertices of $V(G) \setminus (V(K) \cup
\{u,v\})$. Then, given a $c$-pumpkin-model of $H \setminus X$, one can find in
polynomial time a $c$-pumpkin-model of $G \setminus X$.
\end{lemma}
\begin{proof}
Let $\Gamma:= \Gamma(K, u, v)$, $\gamma :=\gamma(K, u, v)$, and $k:= c -
\gamma$. Let $\{A, B\}$ denote the given $c$-pumpkin-model of $H \setminus X$.

If $u \notin A \cup B$ or $v \notin A \cup B$, then $\{A, B\}$ is a
$c$-pumpkin-model in $G \setminus X$ and we are done. Thus, exchanging $A$ and
$B$ if necessary, we may assume that either $u, v \in A$, or $u\in A$ and $v
\in B$. In the first case, since $G[A \cup V(K)]$ is connected, $\{A \cup V(K),
B\}$ is a $c$-pumpkin-model in $G \setminus X$. Now suppose that $u\in A$ and
$v \in B$. We need to consider which types of $c$-pumpkin-models appear in
$\Gamma +_{k} uv$.

If $\Gamma +_{k} uv$ contains a $c$-pumpkin-model $\{A', B'\}$ with $u \in A'$
and $v \in B'$ then there are exactly $\gamma$ edges linking $A'$ to $B'$ in
the graph $\Gamma$, and hence $\{A \cup A', B\cup B'\}$ is a $c$-pumpkin-model
in $G\setminus X$, as desired.

If $\Gamma +_{k} uv$ has a $c$-pumpkin-model $\{A', B'\}$ with $u, v \in A'$,
then $k=1$ and $\gamma = c - 1$. In $H[A \cup B]$ there is a path $P$ linking
$u$ to $v$ that avoids the $c-1$ edges that resulted from the application of
{\bf Z2} on the $c$-outgrowth $(K, u, v)$. (Note that $P$ could possibly
consists of a single edge linking $u$ to $v$.) Then $\{A' \cup V(P), B'\}$ is a
$c$-pumpkin-model in $G\setminus X$.
\end{proof}

Next we show that the converse of the above lemma also holds.

\begin{lemma}
\label{lem:intermediate2}
 Let $c$ be a fixed positive integer. Suppose $H$
is obtained by applying rule {\bf Z2} on a $c$-outgrowth $(K, u, v)$ of a graph
$G$. Let $X$ be an arbitrary subset of vertices of $V(G) \setminus (V(K) \cup
\{u,v\})$. Then, given a $c$-pumpkin-model of $G \setminus X$, one can find in
polynomial time a $c$-pumpkin-model of $H \setminus X$.
\end{lemma}
\begin{proof}
Let $\Gamma:= \Gamma(K, u, v)$, $\gamma :=\gamma(K, u, v)$, and $k:= c -
\gamma$. Let $\{A, B\}$ denote the given $c$-pumpkin-model of $G \setminus X$.
We may assume that this model is minimal (if not, one can obviously make it
minimal in polynomial time).

If $u \notin A \cup B$ or $v \notin A \cup B$, then by minimality of $\{A, B\}$
both $A$ and $B$ avoid $V(K)$. Thus $\{A, B\}$ is a $c$-pumpkin-model in $H
\setminus X$, and we are done. Hence, exchanging $A$ and $B$ if necessary, we
may assume that either $u, v \in A$, or $u\in A$ and $v \in B$. In the second
case, at most $\gamma$ edges between $A$ and $B$ in $G \setminus X$ are
included in $\Gamma$. Since there are $\gamma$ extra edges between $u$ and $v$
in $H$ compared to $G$, it follows that $\{A \setminus V(K), B \setminus
V(K)\}$ is a $c$-pumpkin-model in $H \setminus X$.

Now suppose that $u, v\in A$. If $B \subseteq V(K)$ then all edges between $A$
and $B$ in $G \setminus X$ are in $\Gamma$. Let $P$ be a path in $G[A]$ linking
$u$ to $v$. (Note that the path $P$ possibly consists of a single edge.) Then
$P$ is disjoint from $V(K)$, as otherwise $P \subseteq \Gamma$ and $\{A \cap
V(\Gamma), B\}$ would be a $c$-pumpkin-model in $\Gamma$. Thus in particular
$k=1$ and $\gamma = c - 1$. Since there are $c-1$ extra edges between $u$ and
$v$ in $H$ compared to $G$, and $P$ avoids all these edges, $\{ \{u\}, \{v\}
\cup (V(P) \setminus \{u\}) \}$ is a $c$-pumpkin-model in $H \setminus X$.

If $B \not \subseteq V(K)$ then $B$ is disjoint from $V(K)$. Since $u$ and $v$
are linked by at least $\gamma \geq 1$ edges in $H$, the graph $H[A \setminus
V(K)]$ is connected, and it follows that $\{A \setminus V(K), B\}$ is a
$c$-pumpkin-model in $H \setminus X$.
\end{proof}

\begin{lemma}
\label{lem:valid_covering} Let $c$ be a fixed positive integer. Suppose $H$
results from the application of {\bf Z1} or {\bf Z2} on a graph $G$. Then
$\tau_{c}(G) = \tau_{c}(H)$. Moreover, every $c$-hitting set $X'$ of $H$ is
also a $c$-hitting set of $G$.
\end{lemma}
This lemma implies that an optimal solution to the \textsc{Minimum}
$c$-\textsc{Pumpkin}-\textsc{Hitting} problem on $G$ can be computed given one
for $H$, and similarly that an approximate solution for $G$ can be obtained
from an approximate solution for $H$. This will be used in our approximation
algorithms in Section~\ref{sec:1stAlgo}.\vspace{.2cm}

\begin{proof}[Proof of Lemma~\ref{lem:valid_covering}] First suppose $H$ results from
the application of {\bf Z1} on $G$ with vertex $v$. We trivially have
$\tau_{c}(G) \geq \tau_{c}(H)$. Let $X'$ be a given $c$-hitting set of $H$. If
$X'$ is not a $c$-hitting set of $G$, then $G\setminus X'$ has a
$c$-pumpkin-model; let $\{A, B\}$ be a minimal one. We have $v\in A \cup B$
since otherwise $\{A, B\}$ would be a $c$-pumpkin-model in $H \setminus X'$. By
the minimality of $\{A, B\}$, we must have $A \cup B \subseteq V(K)$ for some
block $K$ of $G$. But then $K$ is a block of $G$ including $v$ and containing a
$c$-pumpkin-minor, contradicting the assumptions of {\bf Z1}.
Therefore $X'$ is a $c$-hitting set of $G$, and
$\tau_{c}(G) \leq \tau_{c}(H)$ also holds, implying $\tau_{c}(G) =
\tau_{c}(H)$.

Now assume $H$ has been obtained by applying {\bf Z2} on $G$ with $c$-outgrowth
$(K, u, v)$, and let $\Gamma:= \Gamma(K, u, v)$.


First we show $\tau_{c}(G) \geq \tau_{c}(H)$. Let $X$ be a minimum $c$-hitting
set of $G$. If $u\in X$ or $v\in X$, then $X$ is trivially a $c$-hitting set of
$H$, so let us assume $u,v\notin X$. Moreover, we may suppose that $X$ has no
vertex in $K$, since otherwise we could replace all such vertices with the
vertex $u$ (or equivalently $v$). Since $X \subseteq V(G) \setminus V(\Gamma)$
and $G \setminus X$ has no $c$-pumpkin-minor, it follows from
Lemma~\ref{lem:intermediate1} that $H \setminus X$ has no $c$-pumpkin-minor
either, that is, $X$ is a $c$-hitting set of $H$. This shows $\tau_{c}(G) \geq
\tau_{c}(H)$.

Now we prove that $\tau_{c}(G) \leq \tau_{c}(H)$ also holds. Here we show that,
given a $c$-hitting set $X'$ of $H$, the set $X'$ is also a $c$-hitting set of
$G$. Hence, this will also prove the second part of the lemma. If $u \in X'$ or
$v \in X'$, then $X'$ is trivially a $c$-hitting set of $G$. If $u, v \notin
X'$, then Lemma~\ref{lem:intermediate2} implies that $G \setminus X'$ has no
$c$-pumpkin-minor, that is, that $X'$ is a $c$-hitting set of $G$. This shows
$\tau_{c}(G) \leq \tau_{c}(H)$, and therefore $\tau_{c}(G) = \tau_{c}(H)$.
\end{proof}

We conclude this section with a lemma similar to Lemma~\ref{lem:valid_covering}
for $c$-packings.

\begin{lemma}
\label{lem:valid_packing} Let $c$ be a fixed positive integer. Suppose $H$
results from the application of {\bf Z1} or {\bf Z2} on a graph $G$. Then
$\nu_{c}(G) = \nu_{c}(H)$. Moreover, given a $c$-packing $\mathcal{M}'$ of $H$
one can compute in polynomial time a $c$-packing $\mathcal{M}$ of $G$ with
$|\mathcal{M}| = |\mathcal{M}'|$.
\end{lemma}
\begin{proof}
First suppose $H$ results from the application of {\bf Z1} on $G$ with vertex
$v$. Clearly, every $c$-packing of $H$ is a $c$-packing for $G$. Thus
$\nu_{c}(G) \geq \nu_{c}(H)$, and it is enough to show the reverse inequality.
Consider a  $c$-packing of $G$. We may assume that every $c$-pumpkin-model in
that packing is minimal. Thus each such model is contained in some block of
$G$, and hence
 avoids the  vertex $v$.
Therefore the packing also exists in $H$, implying $\nu_{c}(G) \leq \nu_{c}(H)$
and $\nu_{c}(G) = \nu_{c}(H)$, as desired.

Now assume $H$ has been obtained by applying {\bf Z2} on $G$ with outgrowth
$(K, u, v)$.

First we show $\nu_{c}(G) \geq \nu_{c}(H)$. Let $\mathcal{M}'=\{M'_{1}, \dots,
M'_{k}\}$ be a given $c$-packing of $H$. We show that a packing of the same
size in $G$ can be computed in polynomial time, which will prove the second
part of the lemma.
If every $M'_{i}$ avoids at least one of $u, v$ then the packing $\mathcal{M}
:= \mathcal{M}'$ is a $c$-packing in $G$ and we are done. So assume one model
in the collection, say without loss of generality $M'_{1}$, includes both $u$
and $v$. Let $X$ be the union of the vertices in $M'_{2}, \dots, M'_{k}$. Since
$M'_{1}$ is a $c$-pumpkin-model in $H\setminus X$, using
Lemma~\ref{lem:intermediate1} we can compute in polynomial time a
$c$-pumpkin-model $M_{1}$ in $G\setminus X$. Hence $\mathcal{M}
:=\{M_{1},M'_{2}, \dots, M'_{k}\}$ is a $c$-packing of the desired size in $G$.

In order to prove $\nu_{c}(G) = \nu_{c}(H)$ it remains to show $\nu_{c}(G) \leq
\nu_{c}(H)$. Let $\{M_{1}, \dots, M_{k}\}$ be a  $c$-packing of $G$. We may
assume that each $M_{i}$ is minimal. Thus if some $M_{i}$ contains some vertex
of $K$ then $M_{i}$ contains both $u$ and $v$. If there is no such model in the
packing then $\{M_{1}, \dots, M_{k}\}$ is also of $c$-packing of $H$ and we are
done. We may thus assume that some model in the packing, say without loss of
generality $M_{1}$, contains both $u$ and $v$. As before, let $X$ be the union
of the vertices in $M_{2}, \dots, M_{k}$. Using Lemma~\ref{lem:intermediate2}
with $M_{1}$ and $X$ we find a $c$-pumpkin-model $M'_{1}$ in $H\setminus X$.
Thus $\{M'_{1},M_{2}, \dots, M_{k}\}$ is a $c$-packing of size $k$ in $H$, as
desired.
\end{proof}

\subsection{Hedgehogs}
\label{sec:hedgehogs}

Recall that a graph is said to be a {\em multipath} if its underlying simple
graph is isomorphic to a path. If $P$ is a multipath and $u,v \in V(P)$, we
write $uPv$ for the subgraph of $P$ induced by the vertices on a $u$--$v$ path
in $P$ (thus edges in $uPv$ have the same multiplicities as in $P$).

A {\em hedgehog} is a pair $(H, P)$, where $H$ is a graph and $P$ is an induced
multipath of $H$ with $|P| \geq 2$ and such that
\begin{itemize}
\item[(i)] the (possibly empty) set $S := V(H) \setminus V(P)$ is a stable set of
$H$; and \vspace{.4mm}
\item[(ii)] every vertex in $S$ has at least two neighbors in $P$.
\end{itemize}
(Let us recall that a {\em stable set} is a set of vertices such that no two of
them are adjacent.)

Consider a hedgehog $(H, P)$. Its {\em size} is defined as $|P|$, the number of
vertices in $P$. A {\em bad cutset} of $(H, P)$ is a set $X=\{u,v\}$ of two
{\em internal} vertices of $P$ such that $H \setminus X$ has a connected
component $K$ avoiding both endpoints of $P$.
This definition is motivated by reduction rule {\bf Z2}: First, if $K$ is such
a component, then $u$ and $v$ each have at least one neighbor in $K$. This is
because either $K$ contains the subpath of $P$ strictly between $u$ and $v$, or
$K$ consists of a unique vertex of $V(H) \setminus V(P)$ which is then adjacent
to $u$ and $v$ (by condition (ii) in the definition of hedgehogs). Hence either
$(K, u, v)$ is a $c$-outgrowth of $H$, or one can find a $c$-pumpkin-minor in
$H[V(K) \cup \{u,v\}]$.

A {\em rooted} $c$-pumpkin-model of $(H, P)$ is a $c$-pumpkin-model $\{A, B\}$
of $H$ with the extra property that $A$ and $B$ both contain an endpoint of
$P$.

Given a hedgehog $(H, P)$ and a connected induced subgraph $Q$ of $P$ with $|Q|
\geq 2$, one can define a hedgehog $(H', Q)$ as follows: First, remove from $H$
every vertex not in $P$ that has no neighbor in $Q$. Then contract every edge
of $P$ not included in $Q$. Finally, remove from the graph every vertex not in
$Q$ that has only one neighbor in $Q$. This defines the graph $H'$. We leave it
to the reader to check that $(H', Q)$ is indeed a hedgehog; we say that $(H',
Q)$ is the {\em contraction} of $(H, P)$ on the multipath $Q$. See
Figure~\ref{fig:hedgehog} for an illustration of this operation. The following
lemma is a direct consequence of the definition.

\begin{figure}
\centering
\includegraphics[width=1\textwidth]{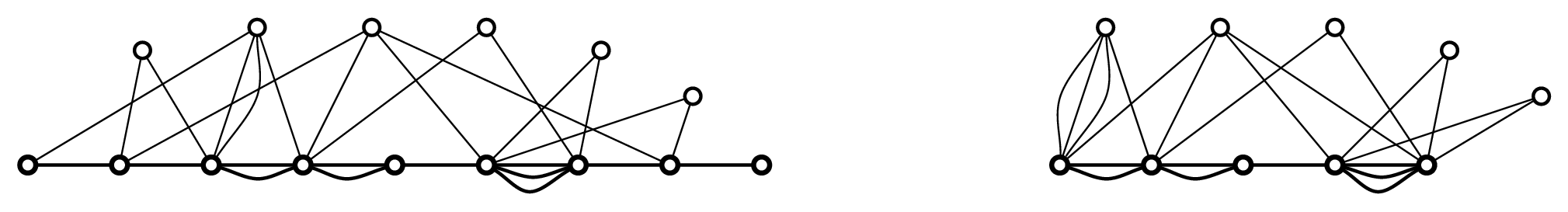}
\caption{A hedgehog $(H, P)$ (left) and a contraction $(H', Q)$ of $(H, P)$
(right). The multipaths $P$ and $Q$ are drawn in bold.\label{fig:hedgehog}}
\end{figure}

\begin{lemma}
\label{lem:hedgehog_cutset} If $(H', Q)$ is a contraction of a hedgehog $(H,
P)$ and $X$ is a bad cutset of $(H', Q)$, then $X$ is also a bad cutset of $(H,
P)$.
\end{lemma}

We show that every big enough hedgehog has a rooted $c$-pumpkin-model or a bad
cutset. This fact will be useful in the subsequent proofs.

\begin{lemma}
\label{lem:hedgehog} Let $c$ be a fixed positive integer. Then every hedgehog
$(H, P)$ of size at least $(2c)^{2c}$ contains a rooted $c$-pumpkin-model or a
bad cutset, either of which can be found in polynomial time.
\end{lemma}
\begin{proof}
The proof is by induction on $c$. The base case $c=1$ is trivial since $P$
directly gives a rooted $1$-pumpkin-model. For the inductive step, assume $c >
1$. Define $f(k)$, for a positive integer $k$, as $f(k) := (2k)^{2k}$. Let $S:=
V(H) \setminus V(P)$. Let $a, b$ be the endpoints of $P$.

If a vertex $v\in S$ has at least $c$ neighbors in $P$, then let $w$ be the
neighbor of $v$ that is closest to $a$ on $P$. Then $A:= V(aPw) \cup \{v\}$ and
$B:= V(P) \setminus A$ both induce a connected subgraph of $H$. Moreover, there
are at least $c-1$ edges from $v$ to $B$, and at least one from $A \setminus
\{v\}$ to $B$ (because of $P$). Since $a\in A$ and $b\in B$, we deduce that
$\{A, B\}$ is a rooted $c$-pumpkin-model of $(H, P)$. Thus we may assume that
every vertex in $S$ has at most $c-1$ neighbors in $P$. In particular we have
$c \geq 3$, since every vertex in $S$ has at least two neighbors in $P$.

The multipath $P$, seen from its endpoint $a$, induces a natural linear
ordering of the neighbors of a given vertex in $S$; we say that two such
neighbors are {\em consecutive} if they are consecutive in that ordering.

Suppose that there exists a vertex $v\in S$ with two consecutive neighbors $x,
y$ such that $|xPy| \geq f(c-1) + 2$. Consider the contraction $(H', Q)$ of
$(H, P)$ on the multipath $Q:= xPy \setminus \{x, y\}$. Since $|Q| \geq
f(c-1)$, by induction $(H', Q)$ has a rooted $(c-1)$-pumpkin-model $\{A' ,
B'\}$ or a bad cutset $X$. If the latter holds, then by
Lemma~\ref{lem:hedgehog_cutset} the set $X$ is also a bad cutset of $(H, P)$
and we are done. Thus we may assume the former holds. In the graph $H$, the
vertex $v$ has no neighbor in $Q$, thus $v$ is not included in $H'$. Hence, we
can obtain a rooted $c$-pumpkin-model $\{A, B\}$ in $(H, P)$ by setting $A:= A'
\cup V(aPx) \cup \{v\}$ and $B:= B' \cup V(yPb)$. Therefore we can assume that,
for every vertex $v \in S$, every two consecutive neighbors of $v$  are at
distance at most $f(c-1)$ on $P$.

Let us enumerate the vertices of $P$ in order as $p_{1}, p_{2}, \dots, p_{k}$,
with $p_{1}=a$ and $p_{k}=b$. We may assume that, for every $i\in \{3, \dots,
k-2\}$,
\begin{equation}
\label{eq:exposed} \textrm{$p_{i}$ is adjacent to some vertex in $S$.}
\end{equation}
Indeed, if not then $\{p_{i-1}, p_{i+1}\}$ would be a bad cutset of $(H, P)$.
Since $k = |P| \geq f(c) \geq f(3) \geq 5$, this implies in particular that $S$
is not empty.

Define an {\em open} interval $I_{v}=(i, j)$ for every vertex $v\in S$, where
$i$ ($j$) is the smallest (largest, respectively) index $t$ such that $p_{t}$
is a neighbor of $v$ in $H$. (Observe that $i < j$ since $v$ has at least two
neighbors.) Now, let $G$ be the interval graph defined by these open intervals,
that is, let $V(G) := S$, and for every two distinct vertices $v, w\in S$, make
$v$ adjacent to $w$ in $G$ if and only if $I_{v} \cap I_{w} \neq \emptyset$.

For a connected subgraph $G'$ of $G$, we define $I(G')$ as the union of the
intervals of vertices in $G'$, that is, $I(G') := \bigcup\{I_{v}: v\in
V(G')\}$. Observe that, since $G'$ is connected, we have $I(G') = (i,j)$ for
some integers $i,j$ with $1 \leq i < j \leq k$.

First suppose that $G$ has at least three connected components. The ordering
$p_{1}, \dots, p_{k}$ of the vertices of $P$ induces an ordering of these
components; let $C$, $C'$, $C''$ be three consecutive connected components in
that ordering. Let $(i, j):= I(C)$, $(i', j'):= I(C')$, and $(i'', j''):=
I(C'')$. Then we have $1 \leq i < j \leq i' < j' \leq i'' < j'' \leq k$, and
every vertex of $S$ that is adjacent to some vertex strictly between $p_{i'}$
and $p_{j'}$ on $P$ has all its neighbors in the set $\{p_{i'}, p_{i' + 1},
\dots, p_{j'}\}$. Thus, for each $w\in V(C')$, the component $K$ of $H -
\{p_{i'}, p_{j'}\}$ that contains $w$ avoids both endpoints of $P$. It follows
that $\{p_{i'}, p_{j'}\}$ is a bad cutset of $(H, P)$. Hence, we may assume
that $G$ has at most two connected components.

Since $G$ has at most two connected components, using~\eqref{eq:exposed} we
deduce that $G$ has a connected component $C$ with $I(C)= (x, y)$ such that
\begin{equation}
\label{eq:Q} y - x + 1 \geq \frac{|P| - 4}{2} \geq \frac{f(c) - 4}{2}.
\end{equation}
Let $Q := p_{x}Pp_{y}$ and let $(H', Q)$ be the contraction of $(H, P)$ on $Q$.
(Note that possibly $Q=P$, in which case $(H', Q)=(H, P)$.) We will show that
$(H', Q)$ contains a rooted $c$-pumpkin-model. The lemma will then follow,
since such a model can be extended straightforwardly to one of $(H, P)$.

First let us observe that $H'$ is an induced subgraph of $H$. This is because,
by our choice of $Q$, every vertex of $S$ that is adjacent to at least two
vertices of $Q$, or to at least one internal vertex of $Q$, has all its
neighbors in $Q$. Let $S' := V(H') \setminus V(Q) = V(C)$. For a vertex $u\in
S'$, let us denote by $\ell(u)$ and $r(u)$ the two integers such that $I_{u} =
(\ell(u), r(u))$.

It follows from our assumptions on $(H, P)$ that, for each $u \in S'$, the
vertex $u$ has at most $c-1$ neighbors in $Q$ and every two consecutive
neighbors of $u$ are at distance at most $f(c-1)$ on $Q$. This implies
\begin{equation}
\label{eq:length} r(u) - \ell(u) \leq (c-2)f(c-1)
\end{equation}
for each $u \in S'$.

In $H'$, the vertices $p_{x}$ and $p_{y}$ each have at least one neighbor in
$S'$. Let $v \in S'$ be a neighbor of $p_{x}$  maximizing $r(v)$, and let $w\in
S'$ be a neighbor of $p_{y}$ minimizing $\ell(w)$. Let $Z$ be a shortest
$v$--$w$ path in the interval graph $G$; enumerate the vertices of $Z$ as
$z_{1}, z_{2}, \dots, z_{m}$ with $z_{1} = v$ and $z_{m}=w$.

By our choice of $v,w$ and the fact that $Z$ is a shortest $v$--$w$ path, we
have
\begin{align}
\label{eq:endpoints1}
\ell(z_{j}) &< \ell(z_{j + 1}) \\
\label{eq:endpoints2} r(z_{j}) &< r(z_{j+1})
\end{align}
for each $j \in \{1, \dots, m-1\}$, and
\begin{equation}
\label{eq:endpoints3} r(z_{j}) \leq \ell(z_{j + 2})
\end{equation}
for each $j \in \{1, \dots, m-2\}$.

Since $I(Z)=I(C)=(x,y)$ we have $y \leq x + \sum_{j=1}^{m} (r(z_{j}) -
\ell(z_{j}))$. Hence $y - x \leq m(c-2)f(c-1)$ by~\eqref{eq:length}.
Using~\eqref{eq:Q} we then obtain
\begin{equation}
m \geq \frac{y - x}{(c-2)f(c-1)} \geq \frac{f(c) - 6}{2(c-2)f(c-1)} =
\frac{(2c)^{2c} - 6}{2(c-2)(2c - 2)^{2c-2}} \geq c.
\end{equation}

Let $d := \lfloor m/2 \rfloor$. Define, for $i\in \{x, \dots, y\}$, the set
$J(p_{i})$ as the set of indices $j\in \{1, \dots, 2d\}$ such that $i \in
\{\ell(z_{j}), r(z_{j})\}$ (let us emphasize that the latter set is not an
interval but just a $2$-element set). We say that $p_{i}$ is a {\em breakpoint} of $Q$ if
$J(p_{i})$ is not empty. (Thus $p_{x}$ is a breakpoint in particular.) It is a
consequence of \eqref{eq:endpoints1}, \eqref{eq:endpoints2}, and
\eqref{eq:endpoints3} that, if $|J(p_{i})| > 1$, then $J(p_{i}) = \{j, j+2\}$
for some $j\in \{1, \dots, 2d - 2\}$. In particular, the numbers in $J(p_{i})$
always have the same parity.

We color the vertices in $V(Q) \cup \{z_{1},\dots, z_{2d}\}$ in {\em black} or
{\em white} as follows. First, for every $j \in \{1, \dots, 2d\}$, color
$z_{j}$ black if $j$ is odd, white if $j$ is even. Next color every breakpoint
$p_{i}$ of $Q$ with the color corresponding to the parity of the numbers in
$J(p_{i})$ (namely, black for odd and white for even). Finally, color every
uncolored vertex of $Q$ with the color of the closest breakpoint of $Q$ in the
direction of $p_{x}$. See Figure~\ref{fig:coloring} for an illustration of the
coloring.

\begin{figure}
\centering
\includegraphics[width=0.65\textwidth]{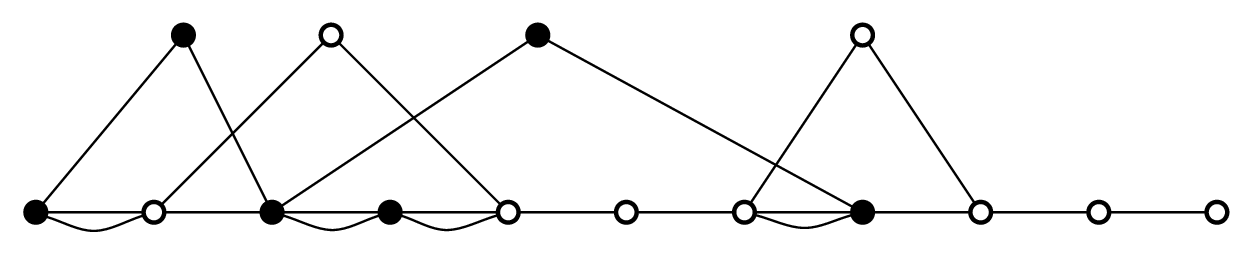}
\caption{Illustration of the coloring for $d=2$. The vertices in
$\{z_{1},\dots, z_{2d}\}$ are on top, the path $Q$ at the bottom, and each
breakpoint $p_{i}$ of $Q$ is linked to the vertices $z_{j}$ such that $j\in
J(p_{i})$. \label{fig:coloring}}
\end{figure}

Let $A$ and $B$ be the set of black and white vertices, respectively. By
construction, $p_{x} \in A$ and $p_{y} \in B$, and each of $A,B$ induces a
connected subgraph of $H'$. Moreover, there are $2d+1 \geq m \geq c$ edges of
$Q$ whose endpoints received distinct colors. It follows that $\{A, B\}$ is a
rooted $c$-pumpkin-model of $(H', Q)$, as desired.

We have shown that $(H, P)$ always has a rooted $c$-pumpkin-model or a bad
cutset. Moreover, it is easily seen from the proof given above that each of
these can be found in polynomial time. This concludes the proof of the lemma.
\end{proof}

We note that no effort has been made in order to optimize the constants in
Lemma~\ref{lem:hedgehog}.

\subsection{Small pumpkins in reduced graphs}
\label{sec:smallPumpkins}

Our goal is to prove that every $n$-vertex $c$-reduced graph $G$ has a
$c$-pumpkin-model of size $\mathcal{O}(\log n)$, where $c$ is a fixed constant.
We will use the following recent result by Fiorini {\it et al}.~\cite{FJTW10}
about the existence of small minors in {\sl simple} graphs with large average
degree.

\begin{theorem}[Fiorini {\it et al}.~\cite{FJTW10}]
\label{thm:smallminors} There is a function $h$ such that every $n$-vertex
simple graph $G$ with average degree at least $2^{t}$ contains a $K_t$-model
with at most $h(t) \cdot \log n$ vertices. Moreover, such a model can be
computed in polynomial time.
\end{theorem}

Since a $K_{t}$-model in a graph directly gives a $c$-pumpkin-model of the same
size for $c = (\lfloor t/2 \rfloor)^{2}$, we have the following corollary from
Theorem~\ref{thm:smallminors}, which is central in the proof of
Lemma~\ref{lem:smallmodelsreducedgraphs}.

\begin{corollary}
\label{cor:smallminors} There is a function $h$ such that every $n$-vertex
simple graph $G$ with average degree at least $2^{2\sqrt{c} + 1}$ contains a
$c$-pumpkin-model with at most $h(c) \cdot \log n$ vertices. Moreover, such a
model can be computed in polynomial time.
\end{corollary}

The next lemma states the existence of small $c$-pumpkin-models in a
$c$-reduced graph $G$. Its proof relies on Lemma~\ref{lem:hedgehog} on
hedgehogs and Corollary~\ref{cor:smallminors}. The proof can be briefly
summarized as follows. An hedgehog in $G$ which is large enough so that
Lemma~\ref{lem:hedgehog} can be applied to it, but at the same time not too
big, witnesses the existence of either a small $c$-pumpkin-model or a
$c$-outgrowth. Hence we may assume that no such hedgehog exists in $G$. The
latter fact is then used to either find directly a small $c$-pumpkin-model, or
a dense-enough minor that is not ``too far'' from $G$ in the sense that it is
obtained by contracting disjoint connected subgraphs of $G$ of bounded radius.
In the latter case, we apply Corollary~\ref{cor:smallminors} on the minor,
yielding a small $c$-pumpkin-model which we then lift back to $G$, incurring
only a constant-factor increase in its size.

\begin{lemma}
\label{lem:smallmodelsreducedgraphs} There is a function $f$ such that every
$n$-vertex $c$-reduced graph $G$ contains a $c$-pumpkin-model of size at most
$f(c)\cdot \log n$. Moreover, such a model can be computed in polynomial time.
\end{lemma}
\begin{proof}
Let
\begin{align*}
k &:= c^{2} \left\lceil 2^{2\sqrt{c} + 1} \right\rceil \\
r &:= (2c)^{2c}k \\
b &:= k^{r}.
\end{align*}

We will prove the lemma with $f$ defined as
$$
f(c) := \max\{krb, 3rc\cdot h(c)\},
$$
where $h$ is the function in Corollary~\ref{cor:smallminors}. Throughout the
proof, a $c$-pumpkin-model of $G$ is said to be {\em small} if it has the
required size, that is, if it has at most $f(c)\log n$ vertices.

Recall that $\mu(G)$ denotes the maximum multiplicity of any edge in $G$.
The lemma trivially holds if $\mu(G) \geq c$, so we may assume $\mu(G) < c$.
Let $W$ be the (possibly empty) subset of vertices of $G$ having degree at
least $k$.

We build a collection $\mathcal{P}$ of vertex-disjoint induced subgraphs of
$G\setminus W$, each isomorphic to a multipath on $r$ vertices. Initially, we
let $\mathcal{P} := \emptyset$ and $G' := G \setminus W$. Then, as long as $G'$
has a connected component $C$ with diameter at least $r-1$, we do the
following: First, we consider two vertices at distance $r-1$ in $C$ and compute
a shortest path $Q$ between these two vertices. Note that the subgraph $P$ of
$G$ induced by $V(Q)$ is a multipath on $r$ vertices. Next, we add $P$ to
$\mathcal{P}$. Finally, we remove from $G'$ the $r$ vertices in $P$.

When the above procedure is finished, every connected component of $G'$ has
diameter less than $r-1$ and maximum degree less than $k$. Hence each such
component has bounded size: at most $k^{r} = b$ vertices. Let $\mathcal{C}$
denote the collection of connected components of $G'$.

An illustration of the sets $W$, $\mathcal{P}$, and $\mathcal{C}$ in the graph
$G$ is given in Figure~\ref{fig:pieces}.

\begin{figure}
\centering
\includegraphics[width=0.7\textwidth]{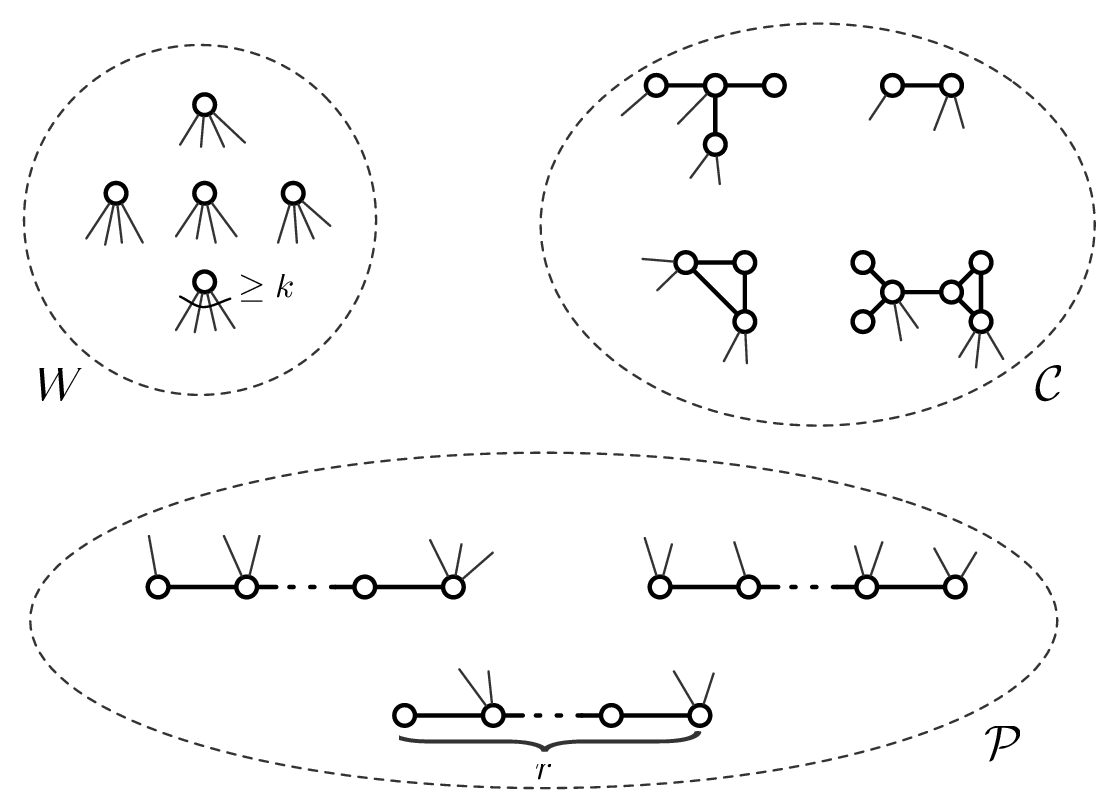}
\caption{The sets $W$, $\mathcal{P}$, and $\mathcal{C}$ in the graph
$G$.\label{fig:pieces}}
\end{figure}

If some subgraph $C \in \mathcal{C}$ contains a $c$-pumpkin-model, then the
size of the model is at most $|C| \leq b \leq f(c) \leq f(c)\log n$, and we are
done. (Note that $n\geq2$ since the model has at least two vertices.) Thus we
may assume  that no subgraph $C \in \mathcal{C}$ contains a $c$-pumpkin-minor.

Let $J$ be the graph obtained from $G$ by contracting each subgraph $C \in
\mathcal{C}$ into a single vertex $v_{C}$. Consider a subgraph $C \in
\mathcal{C}$. We cannot have $\deg^{*}_{J}(v_{C}) = 0$, because otherwise we
could have applied {\bf Z1} on any vertex of $C$ in $G$ (since $C$ has no
$c$-pumpkin-minor). If $\deg^{*}_{J}(v_{C}) = 1$, then let $v$ be an arbitrary
vertex of $C$, and let $w$ be the unique vertex in $V(G) \setminus V(C)$ having
a neighbor in $V(C)$ in the graph $G$. Since {\bf Z1} cannot be applied on $G$
with vertex $v$, there is a block $B$ of $G$ that includes $v$ and containing a
$c$-pumpkin-model. Since $V(B) \subseteq V(C) \cup \{w\}$, this model has size
at most $|B| \leq b + 1 \leq f(c)$, that is, we have found a small
$c$-pumpkin-model of $G$. Therefore, we may assume
\begin{equation}
\label{eq:J2} \deg^{*}_{J}(v_{C}) \geq 2
\end{equation}
for every $C \in \mathcal{C}$.

Let $K$ be the graph obtained from $J$ by contracting each subgraph $P \in
\mathcal{P}$ into a single vertex $v_{P}$. If two vertices of $K$ are linked by
at least $c$ parallel edges (note that these two vertices cannot correspond to two components of $\mathcal{C}$, as no such two components are adjacent), then we directly find a $c$-pumpkin-model in $G$
of size at most $\max\{b + 1, r + 1, 2r, b+r\} \leq f(c)$. Thus we may assume
\begin{equation}
\label{eq:muK} \mu(K) < c.
\end{equation}

We have $\deg^{*}_{K}(v_{C}) \geq 1$ for every $C \in \mathcal{C}$. Let us say
that a subgraph $C \in \mathcal{C}$ is {\em bad} if $\deg^{*}_{K}(v_{C}) = 1$,
and {\em good} otherwise.

We color the vertices of each multipath $P \in \mathcal{P}$ as follows: a
vertex $v\in V(P)$ is colored {\em black} if, in the graph $G$, all its
neighbors outside $P$ belong to bad subgraphs of $\mathcal{C}$; the vertex $v$
is colored {\em white} otherwise. (We remark that $v$ could possibly have no
neighbor outside $P$, in which case $v$ is colored black by our definition.)

\begin{claimN}
\label{claim:multipath} If some multipath $P\in \mathcal{P}$ contains
$(2c)^{2c}$ consecutive black vertices, then one can find a small
$c$-pumpkin-model in $G$.
\end{claimN}
\begin{proof}
Let $Q$ be the subgraph of $P$ induced by these $(2c)^{2c}$ black vertices. Let
$\mathcal{C'}$ be the subset of subgraphs $C \in \mathcal{C}$ such that $v_{C}$
is adjacent to an {\em internal} vertex of $Q$ in the graph $J$. Let $S :=
\{v_{C} : C \in \mathcal{C'}\}$. Since all vertices of $Q$ are colored black,
it follows that internal vertices of $Q$ are only adjacent  in $J$ to vertices
in $V(P) \cup S$, and that every subgraph $C \in \mathcal{C'}$ is bad.

Let $H$ be the graph obtained from $J[V(P) \cup S]$ by contracting every edge
of $P$ not included in $Q$. Since every subgraph $C \in \mathcal{C'}$ is bad,
it follows from \eqref{eq:J2} that, in $H$, every vertex in $S$ has at least
two neighbors in $Q$. Hence $(H, Q)$ is a hedgehog of size $|Q|=(2c)^{2c}$.

The graph $H$ is a minor of the subgraph $G^{*}$ of $G$ induced by
$$
V(P) \cup \bigcup\{V(C): C \in \mathcal{C'}\}.
$$

Since vertices of $P$ have degree at most $k$ in $G$ and $|C| \leq b$ for every
$C \in \mathcal{C'}$, we have
\begin{equation}
\label{eq:Gstar} |G^{*}| \leq r + r(k-1)b \leq f(c).
\end{equation}
We claim that $G^{*}$ contains a $c$-pumpkin-minor. By~\eqref{eq:Gstar}, such a
minor directly yields a small $c$-pumpkin-model of $G$. Arguing by
contradiction, assume $G^{*}$ has no $c$-pumpkin-minor. Thus $H$ has no
$c$-pumpkin-minor either.

Applying Lemma~\ref{lem:hedgehog} on $(H, Q)$, we obtain either a bad cutset
$X$ of $(H, Q)$ or a $c$-pumpkin-model of $H$. The latter case cannot happen
since $H$ has no $c$-pumpkin-minor, so assume the former holds and let $\{u,
v\}:= X$. Consider a connected component $T$ of $H \setminus X$ that avoids
both endpoints of $Q$. Let $Z$ be the subgraph of $G$ induced by $(V(T) \cap
V(Q)) \cup \bigcup\{V(C) : v_{C} \in V(T)\}$. It follows from the definition of
$H$ and our choice of $T$ that $Z$ is a connected component of $G \setminus X$
such that $u$ and $v$ are both adjacent to some vertex in $Z$. Since $G^{*}$
has no $c$-pumpkin-minor, it follows that $(Z, u, v)$ is a $c$-outgrowth of
$G$. But this implies that we could have applied {\bf Z2} on $G$ with the
$c$-outgrowth $(Z, u, v)$, a contradiction.
\end{proof}

By Claim~\ref{claim:multipath}, we may assume that, for every $P\in
\mathcal{P}$, the number $w(P)$ of white vertices in $P$ satisfies
\begin{equation}
\label{eq:wP} w(P) \geq \frac{r}{(2c)^{2c}}=k.
\end{equation}

Our aim now is to use \eqref{eq:wP} to define a minor of $K$ with large minimum
degree. First, for every good subgraph $C\in \mathcal{C}$, ``assign'' $v_{C}$
to an arbitrary neighbor of $v_{C}$ in $K$. Next, for every $w\in W$, contract
all edges $v_{C}w$ of $K$ into the vertex $w$ for all vertices $v_{C}$ assigned
to $w$. Similarly, for every $P\in \mathcal{P}$, contract all edges
$v_{C}v_{P}$ into the vertex $v_{P}$ for all vertices $v_{C}$ assigned to
$v_{P}$. Finally, remove the vertex $v_{C}$ for every bad subgraph $C\in
\mathcal{C}$. The resulting graph is denoted $L$.

For every vertex of $L$ there is a natural induced subgraph of $G$ that
corresponds to it, namely the subgraph defined by all the edges that were
contracted into $w$. Let $S_{w}$ and $S_{P}$ be the (induced) subgraph of $G$
that corresponds to the vertex $w\in W$ and $v_{P}$ ($P \in \mathcal{P}$) of
$L$, respectively. The subgraphs $S_{w}$ ($w\in W$) and $S_{P}$ ($P \in
\mathcal{P}$) of $G$ have diameter at most $2r$ and $3r$, respectively. Thus,
by Lemma~\ref{lem:models_diameter}, a $c$-pumpkin-model of $L$ of size $q$ can
be turned into one of $G$ of size at most $3rc\cdot q$. Hence, in order to
conclude the proof, it is enough to find a $c$-pumpkin-model in $L$ of size at
most $h(c) \log |L|$, since
$$
h(c) \log |L| \leq \frac{f(c)}{3rc} \log |L| \leq \frac{f(c)}{3rc} \log n.
$$
To do so, we will show that $L$ has large minimum degree.

First consider a vertex $w\in W$, and note that $\deg_{K}(w)=\deg_{G}(w) \geq k$. Let $a$ be the number of edges incident with
$w$ in $K$ such that the other endpoint is a vertex of the form $v_{C}$ ($C \in
\mathcal{C}$) that was assigned to $w$. By \eqref{eq:J2}, $w$ cannot be
adjacent in $K$ to a vertex $v_{C}$ corresponding to a bad subgraph $C \in
\mathcal{C}$. Thus, it follows from the definitions of good subgraphs and $L$
that
$$
\deg_{L}(w) \geq \frac{a}{\mu(K)} + (\deg_{K}(w) - a).
$$
(The $\frac{a}{\mu(K)}$ term above comes from the fact that each vertex $v_{C}$
that was assigned to $w$ contributes at least one to the degree of $w$ in $L$,
while in $K$ there were at most $\mu(K)$ edges between $v_{C}$ and $w$.) Using
\eqref{eq:muK} we obtain
\begin{equation}
\label{eq:degL-W} \deg_{L}(w) > \frac{a}{c} + (\deg_{K}(w) - a) \geq
\frac{\deg_{K}(w)}{c} \geq \frac{k}{c}.
\end{equation}

Now consider a multipath $P \in \mathcal{P}$. Let $a'$ be the number of edges
incident with $v_{P}$ in $K$ such that the other endpoint is a vertex of the
form $v_{C}$ ($C \in \mathcal{C}$) that was assigned to $v_{P}$. Let $b'$ be
the number of edges incident with $v_{P}$ in $K$ that are not of the previous
form and also not incident with a vertex $v_{C}$ such that $C$ is bad. By the
definition of white vertices, we have $a' + b' \geq w(P)$ (recall that $w(P)$
is the number of white vertices in $P$). Using~\eqref{eq:wP}, it follows
\begin{equation}
\label{eq:degL-P} \deg_{L}(v_{P}) \geq \frac{a'}{\mu(K)} + b' > \frac{a'}{c} +
b' \geq \frac{w(P)}{c} \geq \frac{k}{c}.
\end{equation}

It follows from \eqref{eq:degL-W} and \eqref{eq:degL-P} that $L$ has minimum
degree at least $k/c$. If $\mu(L) \geq c$, then $L$ has a $c$-pumpkin-model of
size two and we are trivially done, so let us assume $\mu(L) < c$. Then the
underlying simple graph $L'$ of $L$ has minimum degree at least $k / c^{2} \geq
2^{2\sqrt{c} + 1}$. Using Corollary~\ref{cor:smallminors} on $L'$, we find a
$c$-pumpkin-model in $L$ of the desired size, that is, of size at most $h(c)
\log |L|$.

Finally, we note that each step of the proof can easily be realized in
polynomial time. Therefore, a small $c$-pumpkin-model of $G$ can be found in
polynomial time.\end{proof}

\subsection{Algorithmic consequences}
\label{sec:1stAlgo}

Lemma~\ref{lem:smallmodelsreducedgraphs} can be used to obtain
$\mathcal{O}(\log n)$-approximation algorithms for both the \textsc{Minimum}
$c$-\textsc{Pumpkin}-\textsc{Hitting} and the \textsc{Maximum}
$c$-\textsc{Pumpkin}-\textsc{Packing} problems for every fixed $c \geq 1$, as
we now show.

\begin{algorithm}
\caption{\label{algo:packing_covering}A $\mathcal{O}(\log n)$-approximation
algorithm.}

\hspace{-1.4cm}
\begin{minipage}{1.08\textwidth}
\begin{itemize}
\item[] {\bf INPUT:} An arbitrary graph $G$
\item[] {\bf OUTPUT:} A $c$-packing $\mathcal{M}$ of $G$
and a $c$-hitting set $X$ of $G$ s.t.\ $|X| \leq (f(c) \log |G|) \cdot
|\mathcal{M}|$
\item[] $\mathcal{M} \leftarrow \emptyset$; $X \leftarrow \emptyset$
\item[] {\bf If} $|G| \leq 1$: Return $\mathcal{M}$, $X$ $\ \ \ ${\em /* $G$ cannot have a $c$-pumpkin-minor */}
\item[] {\bf Else}:
\begin{itemize}
\item[] {\bf If} $G$ is not $c$-reduced:
\begin{itemize}
\item[] Apply a reduction rule on $G$, giving a graph $H$
\item[] Call the algorithm on $H$, giving a packing $\mathcal{M}'$ and a $c$-hitting set $X'$ of $H$
\item[] Compute using Lemma~\ref{lem:both}(b)
a $c$-packing $\mathcal{M}$ of $G$ with $|\mathcal{M}| = |\mathcal{M}'|$
\item[] Compute using Lemma~\ref{lem:both}(a)
a $c$-hitting set $X$ of $G$ with $|X| \leq |X'|$
\item[] Return $\mathcal{M}$, $X$
\end{itemize}
\item[] {\bf Else}:
\begin{itemize}
\item[] Compute using Lemma~\ref{lem:smallmodelsreducedgraphs} a $c$-pumpkin-model $M=\{A,B\}$
of $G$ with
\item[] \hspace{0.6cm} $|A\cup B| \leq f(c)\log |G|$
\item[] $H \leftarrow G \setminus (A\cup B)$
\item[] Call the algorithm on $H$, giving a packing $\mathcal{M}'$ and a $c$-hitting set $X'$ of $H$
\item[] $\mathcal{M} \leftarrow  \mathcal{M}' \cup \{M\}$
\item[] $X \leftarrow X' \cup A \cup B$
\item[] Return $\mathcal{M}$, $X$
\end{itemize}
\end{itemize}
\end{itemize}
\end{minipage}
\end{algorithm}

\begin{theorem}
\label{thm:logn_packing_covering} Given an $n$-vertex graph $G$, a
$\mathcal{O}(\log n)$-approximation for both the \textsc{Minimum}
$c$-\textsc{Pumpkin}-\textsc{Hitting} and the \textsc{Maximum}
$c$-\textsc{Pumpkin}-\textsc{Packing} problems on $G$ can be computed in
polynomial time using Algorithm~\ref{algo:packing_covering}, for any fixed
integer $c \geq 1$.
\end{theorem}
\begin{proof}Consider Algorithm~\ref{algo:packing_covering}, where $f$ is the
function in Lemma~\ref{lem:smallmodelsreducedgraphs}. We will show that this
algorithm provides a $\mathcal{O}(\log n)$-approximation for the two problems
under consideration.

It should be clear that the collection $\mathcal{M}$ returned by
Algorithm~\ref{algo:packing_covering} is a $c$-packing of $G$, and similarly
that the set $X$ is a $c$-hitting set of $G$. Thus it is enough to show that
they satisfy $|X| \leq (f(c)\log |G|) \cdot |\mathcal{M}|$ as claimed in the
description of the algorithm. Indeed, since $|\mathcal{M}| \leq \nu_{c}(G) \leq
\tau_{c}(G) \leq |X|$ and $f(c)$ is a constant depending only on $c$, this
implies that the approximation factor of Algorithm~\ref{algo:packing_covering}
is $\mathcal{O}(\log n)$ for both the \textsc{Minimum}
$c$-\textsc{Pumpkin}-\textsc{Hitting} and the \textsc{Maximum}
$c$-\textsc{Pumpkin}-\textsc{Packing} problems.

We prove the inequality $|X| \leq (f(c)\log |G|) \cdot |\mathcal{M}|$ by
induction on $|G|$. The inequality is obviously true in the base case, namely
when $|G| \leq 1$, so let us assume $|G| > 1$.

If $G$ is not $c$-reduced, then by induction the packing $\mathcal{M}'$ and the
$c$-hitting set $X'$ of $H$ considered by the algorithm satisfy $|X'| \leq
(f(c)\log |H|) \cdot |\mathcal{M}'|$, and we obtain
$$
|X| \leq |X'| \leq (f(c)\log |H|) \cdot |\mathcal{M}'| = (f(c)\log |H|) \cdot
|\mathcal{M}| \leq (f(c)\log |G|) \cdot |\mathcal{M}|
$$
as desired.

If $G$ is $c$-reduced, then by induction the packing $\mathcal{M}'$ and the
$c$-hitting set $X'$ of $H$ satisfy $|X'| \leq (f(c)\log |H|) \cdot
|\mathcal{M}'|$, and we have
\begin{align*}
|X| &= |X'| + |A\cup B| \\
&\leq (f(c)\log |H|) \cdot |\mathcal{M}'| + f(c)\log |G|  \\
&\leq (f(c)\log |G|) \cdot (|\mathcal{M}'| + 1) \\
&= (f(c)\log |G|) \cdot |\mathcal{M}|.
\end{align*}
Thus $|X| \leq (f(c)\log |G|) \cdot |\mathcal{M}|$ holds in all cases.

Finally, we observe that there are at most $n$ recursive calls during the whole
execution of the algorithm, which implies that its running time is polynomial
in $n$.\end{proof}

\section{Concluding remarks}
\label{sec:concl}

On the one hand, we provided an FPT algorithm running in time
$2^{\mathcal{O}(k)} \cdot n^{\mathcal{O}(1)}$ deciding, for any fixed $c \geq
1$, whether all $c$-pumpkin-models of a given graph can be hit by at most $k$
vertices. In our algorithms we used protrusions
but it may be possible to avoid it by further exploiting the
structure of the graphs during the iterative compression routine (for example,
a graph excluding the $3$-pumpkin is a forest of cacti). We did not focus on
optimizing the constants involved in our algorithms; it may be worth doing it,
as well as
 enumerating all solutions, in the same spirit
of~\cite{GGH+06} for \textsc{Feedback Vertex Set}.

It is natural to ask whether there exist faster algorithms for sparse graphs.
Also, it would be interesting to have lower bounds for the running time of
parameterized algorithms for this problem, in the spirit of those recently
provided in~\cite{LMS11b}. One could as well consider other containment
relations, like topological minor, induced minor, or contraction minor.

A more difficult problem seems to find single-exponential algorithms for the
problem of deleting at most $k$ vertices from a given graph so that the
resulting graph has tree-width bounded by some constant $c$. Note that the case
$c=0$ (resp. $c=1$) corresponds to {\sc $p$-Vertex Cover} (resp.  {\sc
$p$-Feedback Vertex Set}). Very recently, this problem has been solved for
$c=2$~\cite{KPP11}, the cases $c\geq 3$ being still open.
One could also consider the parameterized version of packing disjoint
$c$-pumpkin-models, as it has been done for $c=2$ in~\cite{BTY09}.

On the other hand, we provided a $\mathcal{O}(\log n)$-approximation for the
problems of packing the maximum number of vertex-disjoint $c$-pumpkin-models,
and hitting all $c$-pumpkin-models with the smallest number of vertices. It may
be possible that the hitting version admits a constant-factor approximation; so
far, such an algorithm is only known for $c \leq 3$.

As mentioned in the introduction, for the packing version there is a lower
bound of $\Omega(\log^{1/2 - \varepsilon} n)$ on the approximation ratio (under
reasonable complexity-theoretic assumptions). In fact, this lower bound applies
to both the vertex-disjoint packing and the edge-disjoint
packing~\cite{FrSa11}. For $c=2$, the problem of packing a maximum number of
edge-disjoint cycles admits a $\mathcal{O}(\sqrt{\log n})$-approximation,
whereas up to date $\mathcal{O}(\log n)$ is the best approximation ratio known
for vertex-disjoint cycles~\cite{KNS+07}. Therefore, one might expect to get
better approximation algorithms for packing edge-disjoint $c$-pumpkin-models.

Our algorithms use as subroutines some steps that are only of theoretical interest.
For instance, the FPT algorithm of Section~\ref{sec:algorithm} uses a protrusion replacement rule that involves huge constants, and in the whole paper we repeatedly use Courcelle's theorem~\cite{Courcelle92} to test for the existence of a $c$-pumpkin-model in graphs of bounded treewidth. Turning these steps into routines involving reasonable constants is worth investigating.

Finally, a class of graphs $\mathcal{H}$ has the \emph{Erd\H{o}s-P\'{o}sa
property} if there exists a function $f$ such that, for every integer $k$ and
every graph $G$, either $G$ contains $k$ vertex-disjoint subgraphs each
isomorphic to a graph in $\mathcal{H}$, or there is a set $S \subseteq V(G)$ of
at most $f(k)$ vertices such that $G \setminus S$ has no subgraph in
$\mathcal{H}$. Given a connected graph $H$, let $\mathcal{M}(H)$ be the class
of graphs that can be contracted to $H$. Robertson and Seymour~\cite{RoSe86}
proved that $\mathcal{M}(H)$ satisfies the Erd\H{o}s-P\'{o}sa property if and
only if $H$ is planar. Therefore, for every $c \geq 1$, the class of graphs
that can be contracted to the $c$-pumpkin satisfies the Erd\H{o}s-P\'{o}sa
property. But the best known function $f$ is super-exponential
(see~\cite{Die05}), so it would be interesting to find a better function for
this case. The only known lower bound on $f$ is $\Omega(k \log k)$ when $c \geq
2$, which follows from the $\Omega(k \log k)$ lower bound  given by Erd\H{o}s
and P\'{o}sa in their seminal paper~\cite{ErPo65} for $c=2$.

\subsection*{Acknowledgement.} We would like to thank the anonymous referees for helpful remarks that improved the presentation of the article.

\bibliographystyle{abbrv}
\bibliography{theta}

\end{document}